\documentclass[AMA,STIX1COL]{WileyNJD-v2}
\usepackage{amsmath}
\usepackage{epstopdf}
\articletype{RESEARCH ARTICLE}%

\begin{document}

\title{Control Lyapunov-Barrier Function Based Model Predictive Control for Stochastic Nonlinear Affine Systems 
\thanks{This work was supported by National Natural Science Foundation of China under grant 62073015.}}

\author[1]{Weijiang Zheng}

\author[2]{Bing Zhu}

\address[1]{The Seventh Research Division,\\
Beihang University, Beijing 100191, P.R.China.\\(e-mail: zwj123@buaa.edu.cn; zhubing@buaa.edu.cn).}

\abstract[Abstract]{A stochastic model predictive control (MPC) framework is presented in this paper for nonlinear affine systems with stability and feasibility guarantee. We first introduce the concept of stochastic control Lyapunov-barrier function (CLBF) and provide a method to construct CLBF by combining an unconstrained control Lyapunov function (CLF) and control barrier functions. 
The unconstrained CLF is obtained from its corresponding semi-linear system through dynamic feedback linearization. Based on the constructed CLBF, we utilize sampled-data MPC framework to deal with states and inputs constraints, and to analyze stability of closed-loop systems. Moreover, event-triggering mechanisms are integrated into MPC framework to improve performance during sampling intervals.
The proposed CLBF based stochastic MPC is validated via an obstacle avoidance example.}


\maketitle

\section{Introduction}\label{sec1}

Control Lyapunov function (CLF) plays a significant role in nonlinear control design for stabilization problems. When attention turns to safety critical system, the notion of control barrier function (CBF), which is inspired by CLF, provides safety guarantee while ridding of complicated computation of reachable sets\cite{Ames2019}. In general, there are two main streams to unify CLF and CBF. One adds slack variables in CLF constraint if it conflicts with CBF constraints\cite{Ames2016}, i.e., a tradeoff between stability and safety at the expense of performance, while another constructs control Lyapunov-barrier function (CLBF) \cite{Romdlony2016}. In the construction of CLBF, CLF and CBFs are designed independently and then merged to form CLBF. Therefore, few changes will be made on the design if unsafe regions are altered. For nonlinear affine systems, the linearity of dynamics towards control inputs results in linear constraints of CLF and CBF, which can be efficiently solved in optimization-based control while improving performance, e.g., quadratic program (QP) and MPC. In the work of Wu et al\cite{Wu2019}, a CLBF based MPC strategy is presented to deal with limited control inputs. 

For stochastic system, barrier certificate is introduced in the work of Prajna et al\cite{Prajna2007} to bound the risk of entering unsafe regions with its nonnegative supermartingale property. Santoyo et al\cite{Santoyo2021} give a tighter bound and utilizes sum-of-squares optimization to search for polynomial CBF and control policy. A novel CBF construction that guarantees safety with probability 1 is presented in the work of Clark\cite{Clark2021}. For system with high relative degree, Sarkar et al\cite{Sarkar2020} first introduce definitions of high relative degree stochastic CBF, and Wang et al\cite{Wang2021} extend it to bounded control situation. In the work of Yaghoubi et al\cite{Yaghoubi2020}, CBF is composed with a control Lyapunov like function to achieve goal set reachability rather than stability for nonholonomic systems. However, few attempts are made to construct either stochastic CLBF or MPC setting, which is the motivation of this paper.

Model Predictive Control (MPC) is a powerful technique that addresses finite horizon optimal control problems by employing a receding horizon approach. This method effectively handles constraints, making it suitable for various applications. However, the performance of the closed-loop system can significantly deteriorate due to external disturbances. 
To mitigate this issue, stochastic Model Predictive Control (MPC) leverages the random characteristics of disturbances. 
Recent researches has explored stochastic nonlinear MPC, primarily focusing on discrete-time systems. Several notable approaches have emerged, including tube-based methods designed to handle bounded disturbances \cite{bonzanini2019tube, schluter2020constraint}. Additionally, an RNN-based approach has been proposed to address joint chance constraints \cite{yang2022recurrent}. Other advancements involve output-feedback techniques \cite{messerer2022dual}, stability analysis \cite{santos2019constraint, mcallister2022nonlinear}, all specifically tailored for discrete-time systems.
However, in reality, considerable amount of systems are continuous-time systems.
Comparatively, stochastic nonlinear MPC for continuous-time systems has received relatively less focus \cite{mesbah2016stochastic}.

The main idea of this paper is to develop a stochastic MPC method for stochastic nonlinear affine systems that can handle states and inputs constraints and ensure feasibility and stability.
In continuous-time MPC, control actions are exerted in a sample-and-hold manner, hence we consider stability
analysis for sampled-data system. 
An auxiliary Lyapunov controller is designed to 
ensure stability for stochastic nonlinear sampled-data systems and to approximately estimate  feasibility region.
Slack variables are added to recover feasibility. 
The proposed CLBF based stochastic MPC can provide safety guarantee while deal with stabilization problems. 
Moreover, event-triggering mechanisms can be integrated into MPC framework to improve control performance.
An obstacle avoidance example is presented in simulation section.

Main contributions of this paper are summarized as follows:
\begin{itemize}
    \item A new concept of stochastic CLBF is proposed to ensure safety and stability in stochastic sense.
    \item A framework is provided to construct
    stochastic CLBF through dynamic feedback linearization.
    \item A CLBF based stochastic MPC is designed with feasibility and stability guarantee.
\end{itemize}

The remainder of the paper is summarized as follows. We propose a notion of stochastic control Lyapunov-barrier function in Section \ref{sec2}. In Section \ref{sec3}, an unconstrained CLF, which derived from the corresponding linearized system via dynamic feedback linearization, is composed with CBFs to produce CLBF. The proposed CLBF is applied in an MPC setting in Section \ref{sec4}. A wheeled mobile robot example is presented in Section \ref{sec5} to illustrate design procedure and validation.

\section{Stochastic Control Lyapunov-Barrier Function}\label{sec2}

Consider the following stochastic nonlinear affine system
\begin{eqnarray}\label{eq1}
d\mathbf{x} = \left( f\left( \mathbf{x} \right) + g\left( \mathbf{x} \right)\mathbf{u} \right)dt + \sigma\left( \mathbf{x} \right)dW
\end{eqnarray}
where $\mathbf{x} \in X \subseteq \mathbb{R}^{n}$ is the state, $\mathbf{u} \in U \subseteq \mathbb{R}^{m}$ is the constrained control input and $W$ is an n-dimensional Brownian motion. The functions $f$, $g$ and $\sigma$ are locally Lipschitz continuous with proper dimensions. 
We aim to design the control strategy that stabilizes the closed-loop system while guaranteeing safety in a probabilistic sense, i.e., keeping the states away from the unsafe region with probability 1. Such unsafe region $D$ can be expressed by several disjoint simply-connected open sets $D_i$ that defined by locally Lipschitz functions $h_{i}\left( \mathbf{x} \right)$.

\begin{eqnarray}
D = {\bigcup\limits_{i = 1}^{n_{B}}D_{i}}
\end{eqnarray}

\begin{eqnarray}
D_{i} = \left\{ \mathbf{x} \middle| {h_{i}\left( \mathbf{x} \right) < 0} \right\}
\end{eqnarray}

We first introduce the definitions of control Lyapunov function and control barrier function of stochastic system (\ref{eq1}).
\begin{definition}\label{def1}
\cite{Haddad2019}
A twice differentiable positive-definite function $\left. V:X\rightarrow\mathbb{R} \right.$ is called a stochastic control Lyapunov function of (\ref{eq1}) if it satisfies
\begin{eqnarray}\label{eq4}
\begin{matrix}
{\inf} \\
{u \in U} \\
\end{matrix}\mathcal{L}V\left( \mathbf{x} \right) < 0,\forall\mathbf{x} \in X,\mathbf{x} \neq 0
\end{eqnarray}
where the infinitesimal generator $\mathcal{L}V\left( \mathbf{x} \right)$ is given by
\begin{eqnarray}
\mathcal{L}V\left( \mathbf{x} \right) = L_{f}V\left( \mathbf{x} \right) + L_{g}V\left( \mathbf{x} \right)\mathbf{u} + \frac{1}{2}\mathrm{tr}\left( {\sigma^{T}\left( \mathbf{x} \right)\frac{\partial^{2}V}{\partial\mathbf{x}^{2}}\sigma\left( \mathbf{x} \right)} \right)
\end{eqnarray}
\end{definition}
A stochastic nonlinear system (\ref{eq1}) is stochastically asymptotically stabilizable if and only if there exists a stochastic control Lyapunov function satisfying (\ref{eq4}). If control action $\mathbf{u}$ satisfies $\mathcal{L}V\left( \mathbf{x}_{t} \right) < 0$ all the time, the zero solution $\mathbf{x}_{t} \equiv 0$ is asymptotically stable in probability.

\begin{definition}
A twice differentiable function $B_i:X\rightarrow\mathbb{R}$ is called a stochastic (zero) control barrier function if it satisfies
\begin{eqnarray}
B_{i}\left( \mathbf{x} \right) > 0,\forall\mathbf{x} \in D_{i}
\end{eqnarray}
\begin{eqnarray}
\begin{matrix}
{\inf} \\
{u \in U} \\
\end{matrix}\mathcal{L}B_{i}\left( \mathbf{x} \right) \leq - B_{i}\left( \mathbf{x} \right),\forall\mathbf{x} \in X\backslash D_{i}
\end{eqnarray}
\begin{eqnarray}
\left\{ \mathbf{x} \in X \middle| B_{i}\left( \mathbf{x} \right) \leq 0 \right\} \neq \varnothing
\end{eqnarray}
\end{definition}
Similarly, according to Clark's Theorem 3\cite{Clark2021}, if control action $\mathbf{u}$ satisfies $\mathcal{L}B_{i}\left( \mathbf{x}_{t} \right) \leq - B_{i}\left( \mathbf{x}_{t} \right)$ all the time, then $P\left\{ \mathbf{x} \in X/D_{i} \middle| \mathbf{x}_{0} \in X/D_{i} \right\} = 1$.
We now give the notion of stochastic control Lyapunov-barrier function, which is an extension of deterministic systems\cite{Romdlony2016}\cite{Wu2019}.

\begin{definition}\label{def3}
Given an unsafe region $D\subseteq X$, if there exists a twice differentiable function $W_c:X\rightarrow R$, which has a minimum at the origin, satisfying 
\begin{eqnarray}\label{eq9}
W_{c}\left( \mathbf{x} \right) > 0,\forall\mathbf{x} \in D
\end{eqnarray}
\begin{eqnarray}\label{eq10}
\begin{matrix}
{\inf} \\
{u \in U} \\
\end{matrix}\mathcal{L}W_{c}\left( \mathbf{x} \right) < 0,\forall\mathbf{x} \in X\backslash\left( {D \cup \left\{ 0 \right\}} \right)
\end{eqnarray}
\begin{eqnarray}\label{eq11}
\left\{ {\left. {\mathbf{x} \in X} \right|W_{c}\left( \mathbf{x} \right) \leq 0} \right\} \neq \varnothing
\end{eqnarray}
then the function $W_{c}\left( \mathbf{x} \right)$ is called a stochastic control Lyapunov-barrier function for (\ref{eq1}).
\end{definition}
With slight modifications on the proof of Clark's Theorem 3\cite{Clark2021}, the following proposition provides safety guarantee for closed-loop dynamics if there exists the function $W_{c}\left( \mathbf{x} \right)$.
\begin{proposition}\label{pro1}
Given an unsafe region $D\subseteq X$, if there exists a twice differentiable function $W_c:X\rightarrow R$, which has a minimum at the origin, satisfying
\begin{eqnarray}
W_{c}\left( \mathbf{x} \right) > 0,\forall\mathbf{x} \in D
\end{eqnarray}
\begin{eqnarray}
\begin{matrix}
{\inf} \\
{u \in U} \\
\end{matrix}\mathcal{L}W_{c}\left( \mathbf{x} \right) \leq - W_{c}\left( \mathbf{x} \right),\forall\mathbf{x} \in X\backslash\left( {D_{relaxed} \cup \left\{ 0 \right\}} \right)
\end{eqnarray}
\begin{eqnarray}
\left\{ {\left. {\mathbf{x} \in X} \right|W_{c}\left( \mathbf{x} \right) \leq 0} \right\} \neq \varnothing
\end{eqnarray}
\end{proposition}
where $D_{relaxed} = \left\{ \mathbf{x} \in X \middle| W_{c}\left( \mathbf{x} \right) > 0 \right\}$, and if control action $\mathbf{u}$ satisfies $\mathcal{L}W_{c}\left( \mathbf{x}_{t} \right) \leq - W_{c}\left( \mathbf{x}_{t} \right)$ all the time for $\mathbf{x} \in X\backslash\left( {D_{relaxed} \cup \left\{ 0 \right\}} \right)$, then $P\left\{ \mathbf{x}_{t} \notin D_{relaxed} \middle| \mathbf{x}_{0} \in X\backslash D_{relaxed} \right\} = 1$.
\begin{proof}
First notice that $D\subseteq D_{relaxed}$, then we only need to show that
\begin{eqnarray*}
P\left\{ {\begin{matrix}
{\sup} \\
{t^{\prime} < t} \\
\end{matrix}W_{c}\left( \mathbf{x}_{t^{\prime}} \right) > \epsilon} \right\} < \delta,\forall t > 0,\epsilon > 0,\delta > 0
\end{eqnarray*}
since it follows that
\begin{eqnarray*}
P\left\{ {W_{c}\left( \mathbf{x}_{t^{\prime}} \right) \leq 0,\forall t^{\prime} \in \left\lbrack {0,t} \right\rbrack} \right\} = 1
\end{eqnarray*}
Therefore
\begin{eqnarray*}
P\left\{ {\mathbf{x}_{t} \notin D_{relaxed}} \right\} = {\lim\limits_{t\rightarrow\infty}{P\left\{ {W_{c}\left( \mathbf{x}_{t^{\prime}} \right) \leq 0,\forall t^{\prime} \in \left\lbrack {0,t} \right\rbrack} \right\}}} = 1
\end{eqnarray*}
Let $\theta = \max\left( {W_{c}\left( \mathbf{x}_{0} \right),\frac{- \delta\epsilon}{2t}} \right)$. The initial condition $\mathbf{x}_{0} \in X\backslash D_{relaxed}$ and $W_{c}\left( \mathbf{x}_{0} \right) < 0$ imply that $\theta \leq 0$. Apply Itô formula to the twice differentiable function $W_c$
\begin{eqnarray*}
W_{c}\left( \mathbf{x}_{t} \right) = W_{c}\left( \mathbf{x}_{0} \right) + {\int_{0}^{t}{\mathcal{L}W_{c}\left( \mathbf{x}_{\tau} \right)d\tau}} + {\int_{0}^{t}{\sigma\frac{\partial W_{c}}{\partial\mathbf{x}}dW}}
\end{eqnarray*}
Define a sequence of stopping times $\zeta_i$ and $\eta_i$ for $i=0,1,2,\ldots$
\begin{eqnarray*}
\zeta_{0} = 0,\eta_{0} = \inf\left\{ t \middle| W_{c}\left( \mathbf{x}_{t} \right) < \theta \right\}
\end{eqnarray*}
\begin{eqnarray*}
\zeta_{i} = \inf\left\{ t \middle| W_{c}\left( \mathbf{x}_{t} \right) > \theta,t > \eta_{i - 1} \right\}
\end{eqnarray*}
\begin{eqnarray*}
\eta_{i} = \inf\left\{ t \middle| W_{c}\left( \mathbf{x}_{t} \right) < \theta,t > \zeta_{i - 1} \right\}
\end{eqnarray*}
where $\zeta_i$ and $\eta_i$ are the down-crossings and up-crossings of $W_{c}\left( \mathbf{x}_{t} \right)$ over $\theta$ respectively.
To make use of submartingale inequality, we construct a new random process $U_t$ that only integrates the up-crossings part $\tau\in[\zeta_i,\eta_i]$, namely,
\begin{eqnarray*}
U_{t} = U_{0} + {\sum\limits_{i = 0}^{\infty}\left\lbrack {{\int_{\zeta_{i} \land t}^{\eta_{i} \land t}{- \theta d\tau}} + {\int_{\zeta_{i} \land t}^{\eta_{i} \land t}{\sigma\frac{\partial W_{c}}{\partial\mathbf{x}}dW}}} \right\rbrack}
\end{eqnarray*}
where $U_0=\theta$ and $t_{1} \land t_{2} = \min\left( t_{1},t_{2} \right)$. The following inequality suggests that $U_t$ is a submartingale.
\begin{eqnarray*}
{E\left( U \right.}_{t}\left| U_{0} \right) = U_{0} + E\left( {\sum\limits_{i = 0}^{\infty}\left\lbrack {{\int_{\zeta_{i} \land t}^{\eta_{i} \land t}{- \theta d\tau}} + {\int_{\zeta_{i} \land t}^{\eta_{i} \land t}{\sigma\frac{\partial W_{c}}{\partial\mathbf{x}}dW}}} \right\rbrack} \right) = U_{0} + E\left( {\sum\limits_{i = 0}^{\infty}{\int_{\zeta_{i} \land t}^{\eta_{i} \land t}{- \theta d\tau}}} \right) \geq U_{0}
\end{eqnarray*}
We now prove $W_{c}\left( \mathbf{x}_{t} \right) \leq U_{t}$,$\theta\le U_t$ by induction.

\begin{enumerate}[a.]
\item $t = \zeta_{0}$,
\begin{eqnarray*}
W_{c}\left( \mathbf{x}_{0} \right) \leq \max\left( {W_{c}\left( \mathbf{x}_{0} \right),\frac{- \delta\epsilon}{2t}} \right) = \theta = U_{0}
\end{eqnarray*}

\item\label{b.} $t \in \left( {\zeta_{i},\eta_{i}} \right\rbrack$,
\begin{eqnarray*}
W_{c}\left( \mathbf{x}_{t} \right) = W_{c}\left( \mathbf{x}_{\zeta_i} \right) + {\int_{\zeta_{i}}^{t}{\mathcal{L}W_{c}\left( \mathbf{x}_{\mathbf{\tau}} \right) d\tau}} + {\int_{\zeta_{i}}^{t}{\sigma\frac{\partial W_{c}}{\partial\mathbf{x}}dW}}
\end{eqnarray*}
\begin{eqnarray*}
U_{t} = U_{\zeta_{i}} + {\int_{\zeta_{i}}^{t}{- \theta d\tau}} + {\int_{\zeta_{i}}^{t}{\sigma\frac{\partial W_{c}}{\partial\mathbf{x}}dW}}
\end{eqnarray*}
Suppose the conclusion holds up to $t=\zeta_i$, i.e. $W_{c}\left( \mathbf{x}_{\zeta_i} \right) \leq U_{\zeta_{i}},\theta \leq U_{\zeta_{i}}$. Note that $\mathcal{L}W_{c}\left( \mathbf{x} \right) \leq - W_{c}\left( \mathbf{x} \right)$ holds for $\mathbf{x} \in X\backslash\left( {D_{relaxed} \cup \left\{ 0 \right\}} \right)$, and that $W_{c}\left( \mathbf{x} \right)$ has a minimum at $\mathbf{x} = 0$, which means $\mathbf{x}_{t}=0$ is excluded in this situation, then the constraint $\mathcal{L}W_{c}\left( \mathbf{x}_{t} \right) \leq - W_{c}\left( \mathbf{x}_{t} \right)$ always holds during $t\in\left(\zeta_i,\eta_i\right]$ if $\mathbf{x}_{t} \in X\backslash D_{relaxed}$. By definitions of $\zeta_i$, $W_{c}\left( \mathbf{x}_{t} \right) \geq W_{c}\left( \mathbf{x}_{\zeta_{\mathbf{i}}} \right) = \theta$, hence $\mathcal{L}W_{c}\left( \mathbf{x}_{t} \right) \leq - W_{c}\left( \mathbf{x}_{t} \right) \leq - \theta$, which implies $W_{c}\left( \mathbf{x}_{t} \right) \leq U_{t}$.Therefore $\theta \leq W_{c}\left( \mathbf{x}_{t} \right) \leq U_{t}$.

\item  By definitions of $\eta_i$ and $U_t$, we have $W_{c}\left( \mathbf{x}_{t} \right) \leq W_{c}\left( \mathbf{x}_{\eta_{i}} \right) = \theta$ and $U_t=U_{\eta_i}$. Since $W_{c}\left( \mathbf{x}_{\eta_{i}} \right) \leq U_{\eta_{i}},\theta \leq U_{\eta_{i}}$ by induction in \ref{b.}, we have that $W_{c}\left( \mathbf{x}_{t} \right){\leq W_{c}\left( \mathbf{x}_{\eta_{i}} \right) = \theta \leq U_{\eta_{i}} = U}_{t}$.
\end{enumerate}

Therefore $W_{c}\left( \mathbf{x}_{\mathbf{t}} \right) \leq U_{t}$,$\theta\le U_t$ holds for all $t$.
It follows that
\begin{eqnarray*}
P\left\{ {\begin{matrix}
{\sup} \\
{t^{\prime} < t} \\
\end{matrix}W_{c}\left( \mathbf{x}_{t^{\prime}} \right) > \epsilon} \right\} \leq P\left\{ {\begin{matrix}
{\sup} \\
{t^{\prime} < t} \\
\end{matrix}U_{t^{\prime}} > \epsilon} \right\}
\end{eqnarray*}
We apply the submartingale inequality of Theorem 3.8 in Section 1\cite{Karatzas2012} and get
\begin{eqnarray*}
\epsilon P\left\{ {\begin{matrix}
{\sup} \\
{t^{\prime} < t} \\
\end{matrix}U_{t^{\prime}} > \epsilon} \right\} \leq E\left( {\max\left( {0,U_{t}} \right)} \right) = E\left( U_{t} \right) + E\left( {\max\left( {0,{- U}_{t}} \right)} \right)
\end{eqnarray*}
The conclusion $\theta\le U_t$ implies that $\theta \leq \begin{matrix}
{\min} \\
t \\
\end{matrix}U_{t} \leq - E\left( {\max\left( {0,{- U}_{t}} \right)} \right)$, and hence $E\left(\max\left(0,{-U}_t\right)\right)\le-\theta$.
Moreover, 
\begin{eqnarray*}
E\left( U_{t} \right) = U_{0} + E\left( {\sum\limits_{i = 0}^{\infty}{\int_{\zeta_{i} \land t}^{\eta_{i} \land t}{- \theta d\tau}}} \right) \leq U_{0} - \theta t = \theta - \theta t
\end{eqnarray*}
Since $E\left( {\max\left( {0,U_{t}} \right)} \right) = E\left( U_{t} \right) + E\left( {\max\left( {0,{- U}_{t}} \right)} \right) \leq \theta - \theta t - \theta = - \theta t$, we can obtain
\begin{eqnarray*}
P\left\{ {\begin{matrix}
{\sup} \\
{t^{\prime} < t} \\
\end{matrix}W_{c}\left( \mathbf{x}_{t^{\prime}} \right) > \epsilon} \right\} \leq P\left\{ {\begin{matrix}
{\sup} \\
{t^{\prime} < t} \\
\end{matrix}U_{t^{\prime}} > \epsilon} \right\} \leq \frac{- \theta t}{\epsilon} \leq \frac{\delta\epsilon}{2t}\frac{t}{\epsilon} < \delta
\end{eqnarray*}
which completes the proof.
\end{proof}
The following theorem shows the main property of stochastic control Lyapunov-barrier function.
\begin{theorem}
Given an unsafe region $D\subseteq X$, if there exists a stochastic control Lyapunov-barrier function $W_c:X\rightarrow R$ for system (\ref{eq1}) and if control action $\mathbf{u}_{t}$ satisfies $\mathcal{L}W_{c}\left( \mathbf{x}_{t} \right) < 0$ all the time, then $P\left\{ \mathbf{x}_{t} \notin D \right\} = 1,\forall\mathbf{x}_{0} \in X\backslash D_{relaxed}$ and the zero solution $\mathbf{x}_{t} \equiv 0$ is asymptotically stable in probability.
\end{theorem}
\begin{proof}
It can be observed that $W_{c}\left( \mathbf{x} \right) \leq 0,\forall\mathbf{x} \in X\backslash D_{relaxed}$. Since the definition of $W_{c}\left( \mathbf{x} \right)$ implies $\begin{matrix}
{\inf} \\
{u \in U} \\
\end{matrix}\mathcal{L}W_{c}\left( \mathbf{x} \right) < 0 \leq - W_{c}\left( \mathbf{x} \right),\forall\mathbf{x} \in X\backslash\left( D_{relaxed} \cup \left\{ 0 \right\} \right)$, if control action $\mathbf{u}$ satisfies $\mathcal{L}W_{c}\left( \mathbf{x} \right) \leq - W_{c}\left( \mathbf{x} \right)$ all the time, we have that $P\left\{ {\mathbf{x}_{t} \notin D_{relaxed}} \right\} = 1,\forall\mathbf{x}_{0} \in X\backslash D_{relaxed}$ by Proposition \ref{pro1}, which implies $
P\left\{ {\mathbf{x}_{t} \notin D} \right\} = 1,\forall\mathbf{x}_{0} \in X\backslash D_{relaxed}$.

We now proof $\mathbf{x}_{t} \equiv 0$ is asymptotically stable in probability.
Let
\begin{eqnarray*}
V_{c}\left( \mathbf{x} \right) = W_{c}\left( \mathbf{x} \right) - W_{c}(0)
\end{eqnarray*}
Therefore $V_{c}\left( \mathbf{x} \right)$ is twice differentiable and
\begin{eqnarray*}
V_{c}(0) = 0
\end{eqnarray*}
\begin{eqnarray*}
V_{c}\left( \mathbf{x} \right) > 0,\forall\mathbf{x} \in X,\mathbf{x} \neq 0
\end{eqnarray*}
\begin{eqnarray*}
\begin{matrix}
{\inf} \\
{u \in U} \\
\end{matrix}\mathcal{L}V_{c}\left( \mathbf{x} \right) = \begin{matrix}
{\inf} \\
{u \in U} \\
\end{matrix}\mathcal{L}W_{c}\left( \mathbf{x} \right) < 0,\forall\mathbf{x} \in X\backslash\left( {D_{relaxed} \cup \left\{ 0 \right\}} \right)
\end{eqnarray*}
Then $V_{c}\left( \mathbf{x} \right)$ is a stochastic control Lyapunov function of system (\ref{eq1}). Since constraints on control action $\mathcal{L}V_{c}\left( \mathbf{x}_{t} \right) = \mathcal{L}W_{c}\left( \mathbf{x}_{t} \right) < 0$ hold all the time, the zero solution $\mathbf{x}_{t} \equiv 0$ is asymptotically stable in probability, which completes the proof.
\end{proof}

\section{Design of Stochastic CLF and CBF}\label{sec3}
This section consists of stochastic CLF and CBF design. We will first introduce a stochastic CLF design approach for differentially flat system. Since control barrier function is required to be constant in safe regions to reduce its impact on CLF, a kind of stochastic CBF is presented next. Finally, CLF and CBFs are combined to form a stochastic control Lyapunov-barrier function.

\subsection{Stochastic CLF Design through Dynamic Feedback Linearization}
In this section, we consider the stochastic nonlinear affine system (\ref{eq1}) with $\mathbf{\sigma}\left( \mathbf{x} \right) = diag\left( {\sigma_{1}\left( \mathbf{x} \right),\cdots,\sigma_{n}\left( \mathbf{x} \right)} \right)$. Inspired by CLF design for differentially flat systems\cite{Kubo2020}\cite{Kuga2016}, we can construct CLF in stochastic setting through dynamic feedback linearization. We will show that an unconstrained stochastic CLF (without control constraints) can be constructed via dynamic feedback linearization if we design a quadratic CLF for the linearized system.
Let $\bar{\mathbf{x}} = \left\lbrack {\mathbf{x},~p} \right\rbrack^{T}$, where $p$ is a component of $\mathbf{u}$. The augmented system is given by
\begin{eqnarray}\label{eq15}
d\bar{\mathbf{x}} = \left\lbrack {\bar{\mathbf{f}}\left( \bar{\mathbf{x}} \right) + \bar{\mathbf{g}}\left( \bar{\mathbf{x}} \right)\bar{\mathbf{u}}} \right\rbrack dt + \bar{\mathbf{\sigma}}\left( \bar{\mathbf{x}} \right)dW
\end{eqnarray}
where
\begin{eqnarray*}
\bar{\mathbf{f}}\left( \mathbf{x} \right) = \begin{bmatrix}
{\mathbf{f}\left( \mathbf{x} \right)} \\
0 \\
\end{bmatrix}
\end{eqnarray*}
\begin{eqnarray*}
\bar{\mathbf{g}}\left( \bar{\mathbf{x}} \right) = \begin{bmatrix}
{\mathbf{g}\left( \mathbf{x} \right)} & 0 \\
0 & 1 \\
\end{bmatrix},\bar{\mathbf{u}} = \begin{bmatrix}
\mathbf{u} \\
\dot{p} \\
\end{bmatrix}
\end{eqnarray*}
\begin{eqnarray*}
\bar{\mathbf{\sigma}}\left( \bar{\mathbf{x}} \right) = \begin{bmatrix}
{\mathbf{\sigma}\left( \mathbf{x} \right)} \\
0 \\
\end{bmatrix}
\end{eqnarray*}
We assume there exists a diffeomorphism $\left. \mathbf{\phi}:\mathbb{R}^{n + 1}\rightarrow\mathbb{R}^{n + 1} \right.$ that transforms the augmented system to a semi-linear system
\begin{eqnarray}\label{eq16}
d\mathbf{z} = \left( {\mathbf{A}\mathbf{z} + \mathbf{B}\mathbf{u}^{\prime}} \right)dt + \mathbf{\sigma}^{\prime}\left( \mathbf{z} \right)dW
\end{eqnarray}
where
\begin{eqnarray*}
\mathbf{z} = \mathbf{\phi}\left( \bar{\mathbf{x}} \right) = \mathbf{\alpha}\left( \mathbf{x} \right)p + \mathbf{\beta}\left( \mathbf{x} \right)
\end{eqnarray*}
\begin{eqnarray*}
\mathbf{A}\mathbf{z} + \mathbf{B}\mathbf{u}^{\prime} = \left. {\frac{\partial\mathbf{\phi}}{\partial\bar{\mathbf{x}}}\left( {\bar{\mathbf{f}} + \bar{\mathbf{g}}\bar{\mathbf{u}}} \right)} \right|_{\bar{\mathbf{x}} = \mathbf{\phi}^{- 1}{(\mathbf{z})}} + \begin{bmatrix}
{\frac{1}{2}\mathrm{tr}\left( {\mathbf{\sigma}^{T}\left( \bar{\mathbf{x}} \right)\frac{\partial^{2}\phi_{1}}{\partial{\bar{\mathbf{x}}}^{2}}\mathbf{\sigma}\left( \bar{\mathbf{x}} \right)} \right)} \\
 \vdots \\
{\frac{1}{2}\mathrm{tr}\left( {\mathbf{\sigma}^{T}\left( \bar{\mathbf{x}} \right)\frac{\partial^{2}\phi_{n + 1}}{\partial{\bar{\mathbf{x}}}^{2}}\mathbf{\sigma}\left( \bar{\mathbf{x}} \right)} \right)} \\
\end{bmatrix}_{\bar{\mathbf{x}} = \mathbf{\phi}^{- 1}{(\mathbf{z})}}
\end{eqnarray*}
\begin{eqnarray*}
\mathbf{\sigma}^{\prime}\left( \mathbf{z} \right) = \left. {\frac{\partial\mathbf{\phi}}{\partial\bar{\mathbf{x}}}\bar{\mathbf{\sigma}}\left( \mathbf{x} \right)} \right|_{\bar{\mathbf{x}} = \mathbf{\phi}^{- 1}{(\mathbf{z})}}
\end{eqnarray*}
Note that $\mathbf{Az}$ involves no components of $\bar{\mathbf{u}}$ except $p$, hence $\mathbf{u}^{\prime}$ contains the rest. Assume $p=u_k$ for convenience, and we can divide $\frac{\partial\mathbf{\phi}}{\partial\bar{\mathbf{x}}}\bar{\mathbf{g}}\bar{\mathbf{u}}$ into two parts,
\begin{eqnarray*}
\frac{\partial\mathbf{\phi}}{\partial\bar{\mathbf{x}}}\bar{\mathbf{g}}\bar{\mathbf{u}} = \mathbf{\gamma}\left( \bar{\mathbf{x}} \right)p + \mathbf{B}\mathbf{u}^{\prime}
\end{eqnarray*}
where every component of $\mathbf{u}^{\prime}$ is the linear combination of components of $\bar{\mathbf{u}}$ except $u_k$, namely,
\begin{eqnarray*}
\mathbf{u}^{\prime} = f_{u}\left( \bar{\mathbf{x}} \right)\begin{bmatrix}
u_{1} \\
 \vdots \\
u_{k - 1} \\
u_{k + 1} \\
 \vdots \\
u_{m} \\
\dot{p} \\
\end{bmatrix}
\end{eqnarray*}
with $f_{u} \in \mathbb{R}^{m \times m}$. Since $\mathbf{\phi}$ is a diffeomorphism, $f_u$ has full rank. Therefore, 
\begin{eqnarray*}
\mathbf{A}\mathbf{z} = \frac{\partial\mathbf{\phi}}{\partial\bar{\mathbf{x}}}\bar{\mathbf{f}} + \begin{bmatrix}
{\frac{1}{2}\mathrm{tr}\left( {\mathbf{\sigma}^{T}\left( \bar{\mathbf{x}} \right)\frac{\partial^{2}\phi_{1}}{\partial{\bar{\mathbf{x}}}^{2}}\mathbf{\sigma}\left( \bar{\mathbf{x}} \right)} \right)} \\
 \vdots \\
{\frac{1}{2}\mathrm{tr}\left( {\mathbf{\sigma}^{T}\left( \bar{\mathbf{x}} \right)\frac{\partial^{2}\phi_{n + 1}}{\partial{\bar{\mathbf{x}}}^{2}}\mathbf{\sigma}\left( \bar{\mathbf{x}} \right)} \right)} \\
\end{bmatrix} + \mathbf{\gamma}\left( \bar{\mathbf{x}} \right)p
\end{eqnarray*}
\begin{proposition}
If there exists an unconstrained stochastic control Lyapunov function $\bar{V}\left( \mathbf{z} \right) = \mathbf{z}^{T}\mathbf{P}\mathbf{z}$ for the linearized system (\ref{eq16}), the function $V\left( \mathbf{x} \right) = {\min\limits_{p}{\bar{V}\left( {\mathbf{x},p} \right)}}$, where $\bar{V}\left( {\mathbf{x},p} \right) = \bar{V}\left( \mathbf{z} \right)\left. \hspace{0pt} \right|_{\mathbf{z} = \mathbf{\phi}{(\bar{\mathbf{x}})}}$, will also be an unconstrained stochastic control Lyapunov function for system (\ref{eq1}).
\end{proposition}
\begin{proof}
The proof consists of two steps. We will first illustrate that if we design an unconstrained quadratic control Lyapunov function $\bar{V}\left( \mathbf{z} \right) = \mathbf{z}^{T}\mathbf{P}\mathbf{z}$ for system (\ref{eq16}), $\bar{V}\left( \bar{\mathbf{x}} \right) = \bar{V}\left( \mathbf{z} \right)\left. \hspace{0pt} \right|_{\mathbf{z} = \mathbf{\phi}{(\bar{\mathbf{x}})}}$  will be an unconstrained control Lyapunov function for (\ref{eq15}).
Notice that 
\begin{eqnarray*}
\frac{\partial^{2}\bar{V}}{\partial{\bar{\mathbf{x}}}^{2}} = \frac{\partial\mathbf{\phi}~}{\partial\bar{\mathbf{x}}}^{T}\frac{\partial^{2}\bar{V}}{\partial\mathbf{z}^{2}}\frac{\partial\mathbf{\phi}}{\partial\bar{\mathbf{x}}} + {\sum\limits_{i = 1}^{n + 1}{\frac{\partial\bar{V}}{\partial z_{i}}\frac{\partial^{2}\phi_{i}}{\partial{\bar{\mathbf{x}}}^{2}}}}
\end{eqnarray*}
\begin{eqnarray*}
\mathrm{tr}\left( {{\bar{\mathbf{\sigma}}}^{T}\frac{\partial^{2}\phi_{i}}{\partial{\bar{\mathbf{x}}}^{2}}\bar{\mathbf{\sigma}}} \right) = {\sum\limits_{i = 1}^{n}{{\bar{\mathbf{\sigma}}}_{i}^{T}\frac{\partial^{2}\phi_{i}}{\partial{\bar{\mathbf{x}}}^{2}}{\bar{\mathbf{\sigma}}}_{i}}}
\end{eqnarray*}
where $\bar{\mathbf{\sigma}} = \begin{bmatrix}
{\bar{\mathbf{\sigma}}}_{1} & \cdots & {\bar{\mathbf{\sigma}}}_{n} \\
\end{bmatrix}$, and thus we have
\begin{eqnarray*}
\frac{1}{2}\mathrm{tr}\left( {{\mathbf{\sigma}^{\prime}}^{T}\frac{\partial^{2}\bar{V}}{\partial\mathbf{z}^{2}}\mathbf{\sigma}^{\prime}} \right) = \frac{1}{2}\mathrm{tr}\left( {{\bar{\mathbf{\sigma}}}^{T}\left( {\frac{\partial^{2}\bar{V}}{\partial{\bar{\mathbf{x}}}^{2}} - {\sum\limits_{i = 1}^{n + 1}{\frac{\partial\bar{V}}{\partial z_{i}}\frac{\partial^{2}\phi_{i}}{\partial{\bar{\mathbf{x}}}^{2}}}}} \right)\bar{\mathbf{\sigma}}} \right) = \frac{1}{2}\mathrm{tr}\left( {{\bar{\mathbf{\sigma}}}^{T}\frac{\partial^{2}\bar{V}}{\partial{\bar{\mathbf{x}}}^{2}}\bar{\mathbf{\sigma}}} \right) - \frac{1}{2}\frac{\partial\bar{V}}{\partial\mathbf{z}}\begin{bmatrix}
{\sum\limits_{i = 1}^{n}{{\bar{\mathbf{\sigma}}}_{i}^{T}\frac{\partial^{2}\phi_{1}}{\partial{\bar{\mathbf{x}}}^{2}}{\bar{\mathbf{\sigma}}}_{i}}} \\
 \vdots \\
{\sum\limits_{i = 1}^{n}{{\bar{\mathbf{\sigma}}}_{i}^{T}\frac{\partial^{2}\phi_{n + 1}}{\partial{\bar{\mathbf{x}}}^{2}}{\bar{\mathbf{\sigma}}}_{i}}} \\
\end{bmatrix}
\end{eqnarray*}
which yields
\begin{align*}
\mathcal{L}\bar{V}\left( \bar{\mathbf{x}} \right) & = \frac{\partial\bar{V}}{\partial\bar{\mathbf{x}}}\left( {\bar{\mathbf{f}} + \bar{\mathbf{g}}\bar{\mathbf{u}}} \right) + \frac{1}{2}\mathrm{tr}\left( {{\bar{\mathbf{\sigma}}}^{T}\frac{\partial^{2}\bar{V}}{\partial{\bar{\mathbf{x}}}^{2}}\bar{\mathbf{\sigma}}} \right) \\
& = \frac{\partial\bar{V}}{\partial\mathbf{z}}\frac{\partial\mathbf{\phi}}{\partial\bar{\mathbf{x}}}\left( {\bar{\mathbf{f}} + \bar{\mathbf{g}}\bar{\mathbf{u}}} \right) + \frac{\partial\bar{V}}{\partial\mathbf{z}}\begin{bmatrix}
{\frac{1}{2}\mathrm{tr}\left( {{\bar{\mathbf{\sigma}}}^{T}\frac{\partial^{2}\phi_{1}}{\partial{\bar{\mathbf{x}}}^{2}}\bar{\mathbf{\sigma}}} \right)} \\
 \vdots \\
{\frac{1}{2}\mathrm{tr}\left( {{\bar{\mathbf{\sigma}}}^{T}\frac{\partial^{2}\phi_{n + 1}}{\partial{\bar{\mathbf{x}}}^{2}}\bar{\mathbf{\sigma}}} \right)} \\
\end{bmatrix} + \frac{1}{2}\mathrm{tr}\left( {{\bar{\mathbf{\sigma}}}^{T}\frac{\partial^{2}\bar{V}}{\partial{\bar{\mathbf{x}}}^{2}}\bar{\mathbf{\sigma}}} \right) - \frac{1}{2}\frac{\partial\bar{V}}{\partial\mathbf{z}}\begin{bmatrix}
{\sum\limits_{i = 1}^{n}{{\bar{\mathbf{\sigma}}}_{i}^{T}\frac{\partial^{2}\phi_{1}}{\partial{\bar{\mathbf{x}}}^{2}}{\bar{\mathbf{\sigma}}}_{i}}} \\
 \vdots \\
{\sum\limits_{i = 1}^{n}{{\bar{\mathbf{\sigma}}}_{i}^{T}\frac{\partial^{2}\phi_{n + 1}}{\partial{\bar{\mathbf{x}}}^{2}}{\bar{\mathbf{\sigma}}}_{i}}} \\
\end{bmatrix} \\
& = \frac{\partial\bar{V}}{\partial\mathbf{z}}\left( {\mathbf{A}\mathbf{z} + \mathbf{B}\mathbf{u}^{\prime}} \right) + \frac{1}{2}\mathrm{tr}\left( {{\mathbf{\sigma}^{\prime}}^{T}\frac{\partial^{2}\bar{V}}{\partial\mathbf{z}^{2}}\mathbf{\sigma}^{\prime}} \right) = \mathcal{L}\bar{V}\left( \mathbf{z} \right)
\end{align*}
Note that
\begin{eqnarray*}
L_{\bar{\mathbf{g}}}\bar{V}\left( \bar{\mathbf{x}} \right)\bar{\mathbf{u}} = L_{\mathbf{B}}\bar{V}\left( \mathbf{z} \right)\mathbf{u}^{\prime} + \frac{\partial\bar{V}}{\partial\mathbf{z}}\mathbf{\gamma}\left( \bar{\mathbf{x}} \right)p = L_{\bar{\mathbf{g}}}\bar{V}\left( \bar{\mathbf{x}} \right)\begin{bmatrix}
u_{1} \\
 \vdots \\
u_{m} \\
\dot{p} \\
\end{bmatrix} = L_{\mathbf{B}}\bar{V}\left( \mathbf{z} \right)f_{u}\left( \bar{\mathbf{x}} \right)\begin{bmatrix}
u_{1} \\
 \vdots \\
u_{k - 1} \\
u_{k + 1} \\
 \vdots \\
u_{m} \\
\dot{p} \\
\end{bmatrix} + \frac{\partial\bar{V}}{\partial\mathbf{z}}\mathbf{\gamma}\left( \bar{\mathbf{x}} \right)u_{k}
\end{eqnarray*}
If $\bar{V}\left( \mathbf{z} \right)$ is an unconstrained control Lyapunov function for (\ref{eq16}), $\mathcal{L}\bar{V}\left( \mathbf{z} \right)$ becomes negative when $L_{\mathbf{B}}\bar{V}\left( \mathbf{z} \right) = 0$, which implies $\mathcal{L}\bar{V}\left( \bar{\mathbf{x}} \right) = \mathcal{L}\bar{V}\left( \mathbf{z} \right) < 0$ when $L_{\bar{\mathbf{g}}}\bar{V}\left( \bar{\mathbf{x}} \right) = 0$. Therefore $\bar{V}\left( \bar{\mathbf{x}} \right)$ is an unconstrained control Lyapunov function for (\ref{eq15}).

We now design an unconstrained control Lyapunov function $V\left( \mathbf{x} \right)$. In the rest of the proof, we will discuss the relation between $V\left( \mathbf{x} \right)$ and $\bar{V}\left( \bar{\mathbf{x}} \right)$. Similarly, we will show that $V\left( \mathbf{x} \right)$ is an unconstrained control Lyapunov function for (\ref{eq1}) if $\bar{V}\left( \bar{\mathbf{x}} \right)$ is an unconstrained control Lyapunov function for (\ref{eq15}). The infinitesimal generator of $\bar{V}\left( \bar{\mathbf{x}} \right)$ is given by
\begin{eqnarray*}
\mathcal{L}\bar{V}\left( \bar{\mathbf{x}} \right) = \begin{bmatrix}
\frac{\partial\bar{V}}{\partial\mathbf{x}} & \frac{\partial\bar{V}}{\partial p} \\
\end{bmatrix}\begin{bmatrix}
{\mathbf{f} + \mathbf{g}\mathbf{u}} \\
\dot{p} \\
\end{bmatrix} + \frac{1}{2}\mathrm{tr}\left( {\begin{bmatrix}
{\mathbf{\sigma}\left( \mathbf{x} \right)} \\
0 \\
\end{bmatrix}^{T}\frac{\partial^{2}\bar{V}}{\partial{\bar{\mathbf{x}}}^{2}}\begin{bmatrix}
{\mathbf{\sigma}\left( \mathbf{x} \right)} \\
0 \\
\end{bmatrix}} \right) = \frac{\partial\bar{V}}{\partial\mathbf{x}}\left( {\mathbf{f} + \mathbf{g}\mathbf{u}} \right) + \frac{\partial\bar{V}}{\partial p}\dot{p} + \frac{1}{2}{\sum\limits_{i = 1}^{n}{\mathbf{\sigma}_{i}^{T}\frac{\partial^{2}\bar{V}}{\partial\mathbf{x}^{2}}\mathbf{\sigma}_{i}}}
\end{eqnarray*}
We have that $\mathcal{L}\bar{V}\left( \bar{\mathbf{x}} \right) < 0$ when $L_{\bar{\mathbf{g}}}\bar{V}\left( \bar{\mathbf{x}} \right) = 0$, i.e.,
\begin{eqnarray}\label{eq17}
\mathcal{L}\bar{V}\left( \bar{\mathbf{x}} \right) = \frac{\partial\bar{V}}{\partial\mathbf{x}}\mathbf{f} + \frac{1}{2}{\sum\limits_{i = 1}^{n}{\mathbf{\sigma}_{i}^{T}\frac{\partial^{2}\bar{V}}{\partial\mathbf{x}^{2}}\mathbf{\sigma}_{i}}} < 0
\end{eqnarray}
Therefore, we can construct a control Lyapunov function candidate which satisfies $\frac{\partial\bar{V}}{\partial p} = 0$, or equivalently,
\begin{eqnarray}\label{eq18}
V\left( \mathbf{x} \right) = {\min\limits_{p}{\bar{V}\left( {\mathbf{x},p} \right)}}
\end{eqnarray}
if $V_{1} = \mathbf{\alpha}^{T}\mathbf{P}\mathbf{\alpha} > 0$. We have that
\begin{eqnarray*}
\bar{V}\left( {\mathbf{x},p} \right) = \left( {\mathbf{\alpha}p + \mathbf{\beta}} \right)^{T}\mathbf{P}\left( {\mathbf{\alpha}p + \mathbf{\beta}} \right) = \mathbf{\alpha}^{T}\mathbf{P}\mathbf{\alpha}p^{2} + 2\mathbf{\alpha}^{T}\mathbf{P}\mathbf{\beta}p + \mathbf{\beta}^{T}\mathbf{P}\mathbf{\beta} = V_{1}p^{2} + 2V_{2}p + V_{3}
\end{eqnarray*}
\begin{eqnarray*}
\frac{\partial\bar{V}\left( {\mathbf{x},p} \right)}{\partial p} = 2\mathbf{z}^{T}\mathbf{P}\mathbf{\alpha} = 2\mathbf{\alpha}^{T}\mathbf{P}\left( {\mathbf{\alpha}p + \mathbf{\beta}} \right)
\end{eqnarray*}
and the solution of (\ref{eq18}), denoted by $p_{0}\left( \mathbf{x} \right) = - \left( {\mathbf{\alpha}^{T}\mathbf{P}\mathbf{\alpha}} \right)^{- 1}\mathbf{\alpha}^{T}\mathbf{P}\mathbf{\beta}$, is unique, which yields
\begin{eqnarray}\label{eq19}
V\left( \mathbf{x} \right) = \left( {- \mathbf{\alpha}\left( {\mathbf{\alpha}^{T}\mathbf{P}\mathbf{\alpha}} \right)^{- 1}\mathbf{\alpha}^{T}\mathbf{P}\mathbf{\beta} + \mathbf{\beta}} \right)^{T}\mathbf{P}\left( {- \mathbf{\alpha}\left( {\mathbf{\alpha}^{T}\mathbf{P}\mathbf{\alpha}} \right)^{- 1}\mathbf{\alpha}^{T}\mathbf{P}\mathbf{\beta} + \mathbf{\beta}} \right) = - \frac{V_{2}^{2}}{V_{1}} + V_{3}
\end{eqnarray}
We then show that (\ref{eq19}) is an unconstrained control Lyapunov function for system (\ref{eq1}) by pointing out that $\mathcal{L}V\left( \mathbf{x} \right) = \frac{\partial V}{\partial\mathbf{x}}\mathbf{f} + \frac{1}{2}{\sum\limits_{i = 1}^{n}{\mathbf{\sigma}_{i}^{T}\frac{\partial^{2}V}{\partial\mathbf{x}^{2}}\mathbf{\sigma}_{i}}} < 0$ when $L_{\mathbf{g}}V\left( \mathbf{x} \right) = 0$.

We have that 
\begin{eqnarray*}
\frac{\partial\bar{V}}{\partial\mathbf{x}} = \begin{bmatrix}
{p^{2}\frac{\partial V_{1}}{\partial x_{1}} + 2p\frac{\partial V_{2}}{\partial x_{1}} + \frac{\partial V_{3}}{\partial x_{1}}} \\
 \vdots \\
{p^{2}\frac{\partial V_{1}}{\partial x_{n}} + 2p\frac{\partial V_{2}}{\partial x_{n}} + \frac{\partial V_{3}}{\partial x_{n}}} \\
\end{bmatrix}^{T},\frac{\partial V}{\partial\mathbf{x}} = \begin{bmatrix}
{\frac{V_{2}^{2}}{V_{1}^{2}}\frac{\partial V_{1}}{\partial x_{1}} - \frac{2V_{2}}{V_{1}}\frac{\partial V_{2}}{\partial x_{1}} + \frac{\partial V_{3}}{\partial x_{1}}} \\
 \vdots \\
{\frac{V_{2}^{2}}{V_{1}^{2}}\frac{\partial V_{1}}{\partial x_{n}} - \frac{2V_{2}}{V_{1}}\frac{\partial V_{2}}{\partial x_{n}} + \frac{\partial V_{3}}{\partial x_{n}}} \\
\end{bmatrix}^{T}
\end{eqnarray*}
\begin{eqnarray*}
\left( \frac{\partial^{2}\bar{V}}{\partial\mathbf{x}^{2}} \right)_{ij} = p^{2}\frac{\partial^{2}V_{1}}{\partial x_{i}\partial x_{j}} + 2p\frac{\partial^{2}V_{2}}{\partial x_{i}\partial x_{j}} + \frac{\partial^{2}V_{3}}{\partial x_{i}\partial x_{j}}
\end{eqnarray*}
\begin{eqnarray*}
\left( \frac{\partial^{2}V}{\partial\mathbf{x}^{2}} \right)_{ij} = \frac{V_{2}^{2}}{V_{1}^{2}}\frac{\partial^{2}V_{1}}{\partial x_{i}\partial x_{j}} - \frac{2V_{2}}{V_{1}}\frac{\partial^{2}V_{2}}{\partial x_{i}\partial x_{j}} + \frac{\partial^{2}V_{3}}{\partial x_{i}\partial x_{j}} - \frac{2}{V_{1}}\left( {\frac{V_{2}}{V_{1}}\frac{\partial V_{1}}{\partial x_{i}} - \frac{\partial V_{2}}{\partial x_{i}}} \right)\left( {\frac{V_{2}}{V_{1}}\frac{\partial V_{1}}{\partial x_{j}} - \frac{\partial V_{2}}{\partial x_{j}}} \right)
\end{eqnarray*}
It follows that
\begin{eqnarray*}
\frac{\partial V}{\partial\mathbf{x}} = \left. \frac{\partial\bar{V}}{\partial\mathbf{x}} \right|_{p = p_{0}{(\mathbf{x})}}
\end{eqnarray*}
\begin{eqnarray*}
\left( \frac{\partial^{2}V}{\partial\mathbf{x}^{2}} \right)_{ii} = \frac{V_{2}^{2}}{V_{1}^{2}}\frac{\partial^{2}V_{1}}{\partial x_{i}^{2}} - \frac{2V_{2}}{V_{1}}\frac{\partial^{2}V_{2}}{\partial x_{i}^{2}} + \frac{\partial^{2}V_{3}}{\partial x_{i}^{2}} - \frac{2}{V_{1}}\left( {\frac{V_{2}}{V_{1}}\frac{\partial V_{1}}{\partial x_{i}} - \frac{\partial V_{2}}{\partial x_{i}}} \right)^{2} \leq \frac{V_{2}^{2}}{V_{1}^{2}}\frac{\partial^{2}V_{1}}{\partial x_{i}^{2}} - \frac{2V_{2}}{V_{1}}\frac{\partial^{2}V_{2}}{\partial x_{i}^{2}} + \frac{\partial^{2}V_{3}}{\partial x_{i}^{2}} = \left. \left( \frac{\partial^{2}\bar{V}}{\partial\mathbf{x}^{2}} \right)_{ii} \right|_{p = p_{0}{(\mathbf{x})}}
\end{eqnarray*}
Since $L_{\bar{\mathbf{g}}}\bar{V}\left( \bar{\mathbf{x}} \right) = 0$ implies $L_{\mathbf{g}}V\left( \mathbf{x} \right) = 0$, (\ref{eq17}) will hold when $L_{\mathbf{g}}V\left( \mathbf{x} \right) = 0$. Therefore
\begin{eqnarray*}
\frac{\partial V}{\partial\mathbf{x}}\mathbf{f} + \frac{1}{2}{\sum\limits_{i = 1}^{n}{\mathbf{\sigma}_{i}^{T}\frac{\partial^{2}V}{\partial\mathbf{x}^{2}}\mathbf{\sigma}_{i}}} \leq \left. {\frac{\partial\bar{V}}{\partial\mathbf{x}}\mathbf{f} + \frac{1}{2}{\sum\limits_{i = 1}^{n}{\mathbf{\sigma}_{i}^{T}\frac{\partial^{2}\bar{V}}{\partial\mathbf{x}^{2}}\mathbf{\sigma}_{i}}}} \right|_{p = p_{0}{(\mathbf{x})}} < 0
\end{eqnarray*}
and hence (\ref{eq19}) is an unconstrained control Lyapunov function for system (\ref{eq1}).
\end{proof}

\subsection{Stochastic Control Barrier Function}
In the work of Romdlony et al\cite{Romdlony2016}, a lower-bounded CBF that fits the construction of CLBF is presented. However, such method is no longer valid for stochastic systems deal to the twice differentiable condition. Assume a CBF $B_{i}\left( \mathbf{x} \right)$ takes constant values in $X\backslash X_{i}$, where $X_i$ are compact and connected sets satisfying $D_{i} \subseteq X_{i}$ and $D_{relaxed} \subseteq X_{0} = {\bigcup\limits_{i = 1}^{n_{B}}X_{i}} \subseteq X$. Therefore (\ref{eq20}) and (\ref{eq21}) should be satisfied for the twice differentiability on the boundary of $X_i$. 
\begin{eqnarray}\label{eq20}
\frac{\partial B_{i}}{\partial\mathbf{x}} = 0,\forall\mathbf{x} \in \partial X_{i}
\end{eqnarray}
\begin{eqnarray}\label{eq21}
\frac{\partial^{2}B_{i}}{\partial\mathbf{x}^{2}} = 0,\forall\mathbf{x} \in \partial X_{i}
\end{eqnarray}
We assume there exists a function $
F_{Bi}\left( \mathbf{x} \right) \geq 0$ for every $D_i$ such that $D_i$ and $X_i$ can be rewritten as
\begin{eqnarray*}
D_{i} = \left\{ \mathbf{x} \middle| F_{Bi}\left( \mathbf{x} \right) - l_{Di} < 0 \right\}
\end{eqnarray*}
\begin{eqnarray*}
X_{i} = \left\{ \mathbf{x} \middle| F_{Bi}\left( \mathbf{x} \right) - l_{Xi} < 0 \right\}
\end{eqnarray*}
Consider the following function
\begin{eqnarray}\label{eq24}
B_{i}\left( \mathbf{x} \right) = \left\{ \begin{matrix}
{B_{imin} + \frac{B_{imax} - B_{imin}}{1 + e^{- \frac{k_{Bi}{({l_{Di} - F_{Bi}})}}{F_{Bi}{({l_{Xi} - F_{Bi}})}}}~},\forall x \in X_{i}} \\
{B_{imin},\forall x \in X\backslash X_{i}} \\
\end{matrix} \right.
\end{eqnarray}
where $k_{Bi}\left( \mathbf{x} \right) > 0$ is a parameter that adjusts the shape of CBF, and we will detail its selection through an example in Section \ref{sec5}.
We have
\begin{eqnarray*}
\frac{\partial B_{i}}{\partial\mathbf{x}} = - \frac{B_{imax} - B_{imin}}{\left( {1 + e_{1i}~} \right)^{2}}e_{1i}\left\lbrack {\frac{k_{Bi}e_{2i}}{e_{3i}}\frac{\partial F_{Bi}}{\partial\mathbf{x}} - e_{4i}\frac{\partial k_{Bi}}{\partial\mathbf{x}}} \right\rbrack
\end{eqnarray*}
\begin{align*}
\frac{\partial^{2}B_{i}}{\partial\mathbf{x}^{2}} & = \left\{ {e_{5i}\left( \frac{k_{Bi}e_{2i}}{e_{3i}} \right)^{2} - \frac{B_{imax} - B_{imin}}{\left( {1 + e_{1i}~} \right)^{2}}e_{1i}\frac{k_{Bi}}{e_{3i}^{2}}\left\lbrack {2e_{3i}\left( {F_{Bi} - l_{Di}} \right) - 2e_{2i}F_{Bi}\left( {F_{Bi} - l_{Xi}} \right)\left( {2F_{Bi} - l_{Xi}} \right)} \right\rbrack} \right\}\frac{\partial F_{Bi}}{\partial\mathbf{x}}^{T}\frac{\partial F_{Bi}}{\partial\mathbf{x}} \\
& - \left\lbrack {e_{5i}\frac{2k_{Bi}e_{2i}e_{4i}}{e_{3i}} + \frac{B_{imax} - B_{imin}}{\left( {1 + e_{1i}~} \right)^{2}}e_{1i}\frac{e_{2i}}{e_{3i}}} \right\rbrack\frac{\partial F_{Bi}}{\partial\mathbf{x}}^{T}\frac{\partial k_{Bi}}{\partial\mathbf{x}} + e_{5i}e_{4i}^{2}\frac{\partial k_{Bi}}{\partial\mathbf{x}}^{T}\frac{\partial k_{Bi}}{\partial\mathbf{x}} - \frac{B_{imax} - B_{imin}}{\left( {1 + e_{1i}~} \right)^{2}}e_{1i}\frac{k_{Bi}}{e_{3i}}e_{2i}\frac{\partial^{2}F_{Bi}}{\partial\mathbf{x}^{2}} \\
& + \frac{B_{imax} - B_{imin}}{\left( {1 + e_{1i}~} \right)^{2}}e_{1i}e_{4i}\frac{\partial^{2}k_{Bi}}{\partial\mathbf{x}^{2}}
\end{align*}
where 
\begin{eqnarray*}
e_{1i} = e^{- \frac{k_{Bi}{({l_{Di} - F_{Bi}})}}{F_{Bi}{({l_{Xi} - F_{Bi}})}}}, e_{2i} = F_{Bi}^{2} - 2l_{Di}F_{Bi} + l_{Di}l_{Xi}, e_{3i} = F_{Bi}^{2}\left( {l_{Xi} - F_{Bi}} \right)^{2}
\end{eqnarray*}
\begin{eqnarray*}
e_{4i} = \frac{l_{Di} - F_{Bi}}{F_{Bi}\left( {l_{Xi} - F_{Bi}} \right)}, e_{5i} = \frac{2\left( {B_{imax} - B_{imin}} \right)}{\left( {1 + e_{1i}~} \right)^{3}}\left( e_{1i} \right)^{2} - \frac{B_{imax} - B_{imin}}{\left( {1 + e_{1i}~} \right)^{2}}e_{1i}
\end{eqnarray*}
If both $k_{Bi}$ and $F_{Bi}$ are twice differentiable on the boundary of $X_i$, (\ref{eq20}) and (\ref{eq21}) hold. Moreover,
\begin{eqnarray*}
B_{i}\left( \mathbf{x} \right) = \left\{ \begin{matrix}
{B_{imax},F_{Bi}\left( \mathbf{x} \right) = 0} \\
{\frac{B_{imax} + B_{imin}}{2},F_{Bi}\left( \mathbf{x} \right) = l_{Di}} \\
{B_{imin},F_{Bi}\left( \mathbf{x} \right) = l_{Xi}} \\
\end{matrix} \right.
\end{eqnarray*}
Then (\ref{eq24}) is a stochastic control barrier function if $B_{imin}<0$,$B_{imax}>0$,$B_{imax}+B_{imin}>0$.

\subsection{Independent Design of Stochastic CLF and CBFs}
In general, we prefer to design control Lyapunov function and control barrier functions independently and then combine them to form a control Lyapunov-barrier function. Hence control barrier function is required to reach its lower bound and become a constant value when leaving away from the neighborhood of unsafe sets. Proposition \ref{pro3} provides a method to construct a stochastic CLBF by unifying CBFs and an unconstrained CLF. Note that (\ref{eq4}) in Definition \ref{def1} implies the existence of a candidate control input, therefore it is replaced by (\ref{eq10}) and we only need to consider an unconstrained CLF when constructing CLBF. The existence of a candidate control input (feasibility) will be discussed in Section \ref{sec4}. 
\begin{proposition}\label{pro3}
For a given unsafe region $D = {\bigcup\limits_{i = 1}^{n_{B}}D_{i}} \subseteq X$, suppose there exists an unconstrained stochastic control Lyapunov function $V:X\rightarrow R$, which satisfies $\mathcal{L}V\left( \mathbf{x} \right) < 0$ when $L_{\mathbf{g}}V\left( \mathbf{x} \right) = 0$ instead of (\ref{eq4}), and stochastic control barrier functions $\left. B_{i}:X\rightarrow R,i = 1,2,\ldots,n_{B} \right.$ of system (\ref{eq1}), which satisfy
\begin{eqnarray}
c_{1}\left\| \mathbf{x} \right\|^{2} \leq V\left( \mathbf{x} \right) \leq c_{2}\left\| \mathbf{x} \right\|^{2},\forall\mathbf{x} \in R^{n},c_{2} > c_{1} > 0
\end{eqnarray}
\begin{eqnarray}
B_{i}\left( \mathbf{x} \right) = - \eta_{i} < 0,\forall\mathbf{x} \in X\backslash X_{i},B_{i}\left( \mathbf{x} \right) \geq - \eta_{i},\forall\mathbf{x} \in X_{i}
\end{eqnarray}
where $X_i$ are compact and connected sets satisfying $D_i\subseteq X_i$ and $D_{relaxed}\subseteq X_0=\bigcup_{i=1}^{n_B}X_i\subseteq X$. Then the following function $W_{c}\left( \mathbf{x} \right)$ is a stochastic control Lyapunov-barrier function if (\ref{eq10}) holds. 

\begin{eqnarray*}
W_{c}\left( \mathbf{x} \right) = V\left( \mathbf{x} \right) + {\sum\limits_{i = 1}^{n_{B}}{\lambda_{i}B_{i}\left( \mathbf{x} \right)}} + \kappa
\end{eqnarray*}
where
\begin{eqnarray*}
\lambda_{i} > \frac{c_{2}c_{3i} - c_{1}c_{4i}}{\eta_{i}}
\end{eqnarray*}
\begin{eqnarray*}
{\sum\limits_{j = 1,j \neq i}^{n_{B}}{\lambda_{j}\eta_{j}}} - c_{1}c_{4i} < \kappa < {\sum\limits_{j = 1}^{n_{B}}{\lambda_{j}\eta_{j}}} - c_{2}c_{3i}
\end{eqnarray*}
\begin{eqnarray*}
c_{3i} = \begin{matrix}
{\max} \\
{x \in \partial X_{i}} \\
\end{matrix}\left\| \mathbf{x} \right\|^{2},c_{4i} = \begin{matrix}
{\min} \\
{x \in D_{i}} \\
\end{matrix}\left\| \mathbf{x} \right\|^{2}
\end{eqnarray*}
\begin{proof}
Since $V\left( \mathbf{x} \right)$ and $B_{i}\left( \mathbf{x} \right)$ are twice differentiable and have a minimum at the origin respectively, $W_{c}\left( \mathbf{x} \right)$ is also a twice differentiable function with a minimum at the origin.
For $\mathbf{x} \in D_{i}$,
\begin{eqnarray*}
W_{c}\left( \mathbf{x} \right) = V\left( \mathbf{x} \right) + \lambda_{i}B_{i}\left( \mathbf{x} \right) - {\sum\limits_{j = 1,j \neq i}^{n_{B}}{\lambda_{j}\eta_{j}}} + \kappa \geq c_{1}\left\| \mathbf{x} \right\|^{2} - {\sum\limits_{j = 1,j \neq i}^{n_{B}}{\lambda_{j}\eta_{j}}} + \kappa \geq c_{1}c_{4i} - {\sum\limits_{j = 1,j \neq i}^{n_{B}}{\lambda_{j}\eta_{j}}} + \kappa > 0
\end{eqnarray*}
we have that (\ref{eq9}) holds.
For $\mathbf{x} \in \partial X_{i}$,
\begin{eqnarray*}
W_{c}\left( \mathbf{x} \right) \leq c_{2}\left\| \mathbf{x} \right\|^{2} - {\sum\limits_{i = 1}^{n_{B}}{\lambda_{i}\eta_{i}}} + \kappa \leq c_{2}c_{3i} - {\sum\limits_{i = 1}^{n_{B}}{\lambda_{i}\eta_{i}}} + \kappa < 0
\end{eqnarray*}
and hence (\ref{eq11}) holds. If (\ref{eq10}) holds, $W_{c}\left( \mathbf{x} \right)$ is a stochastic control Lyapunov-barrier function by Definition \ref{def3}.
\end{proof}
\end{proposition}

\section{Control Lyapunov-Barrier Function Based Stochastic Model Predictive Control}\label{sec4}
We now apply the proposed stochastic CLBF to a model predictive control setting. Consider a stochastic model predictive control problem
\begin{gather}
\mathbf{u}^{*}=
\begin{matrix}
{\arg \min} \\
{\mathbf{u}} \\
\end{matrix}
\left\lbrack {\sum\limits_{i = 0}^{N - 1}{\mathbf{x}^{T}\left( i \middle| k \right)\mathbf{Q}\mathbf{x}\left( i \middle| k \right) + \mathbf{u}^{T}\left( i \middle| k \right)\mathbf{R}\mathbf{u}\left( i \middle| k \right)}
+ R_{2}\delta^{2}\left( i \middle| k \right)} \right\rbrack
\\
s.t.\quad\mathbf{x}\left( {i + 1} \middle| k \right) - \mathbf{x}\left( i \middle| k \right) = \left\lbrack {\mathbf{f}\left( {\mathbf{x}\left( i \middle| k \right)} \right) + \mathbf{g}\left( {\mathbf{x}\left( i \middle| k \right)} \right)\mathbf{u}\left( i \middle| k \right)} \right\rbrack T
\\
\mathbf{x}\left( 0 \middle| k \right) = \mathbf{x}(k)
\\ \label{eq28}
\mathbf{u}\left( i \middle| k \right) \in U
\\ \label{eq29}
\mathcal{L}W_{c}\left( {\mathbf{x}\left( i \middle| k \right),\mathbf{u}\left( i \middle| k \right)} \right) \leq \mathcal{L}W_{c}\left( {\mathbf{x}\left( i \middle| k \right),\mathbf{\phi}\left( \mathbf{x}\left( i \middle| k \right) \right)} \right)+\delta\left( i \middle| k \right),~\forall\mathbf{x} \in X\backslash\left( D \cup \left\{ 0 \right\} \right)
\end{gather}
where $N$ is predictive horizon; $\mathbf{Q}$ and $\mathbf{R}$ are weighting matrices of predicted states and inputs respectively; $R_2$ is the penalty for feasibility recovery; $\mathbf{x}\left( i \middle| k \right)$ denotes the predicted state of nominal system at $t=t_k+iT$; $\mathbf{u}\left( i \middle| k \right)$ denotes predicted input during $t\in[t_k+iT,t_k+iT+T)$. The control input of closed-loop system is $\mathbf{u}(t) = \mathbf{u}\left( 0 \middle| k \right),t \in \lbrack t_k,t_{k+1})$.
$\mathbf{\phi}\left( \mathbf{x}\left( i \middle| k \right) \right)$ is the auxiliary Lyapunov controller and influences performance during sampling intervals. 
We add $\delta\left( i \middle| k \right)$ to recover feasibility when auxiliary Lyapunov controller $\mathbf{\phi}\left( \mathbf{x}\left( i \middle| k \right) \right)$ is not a feasible solution.

The constraints in the optimizations of MPC are (\ref{eq28}) and (\ref{eq29}). Note that constraints in (\ref{eq29}) are associated with admissible control $\mathbf{u}(t) \in U$, thus the region that admits a feasible solution is restricted. In Lyapunov based control, a universal formula\cite{Lin1991} is applied to construct a bounded control law as an auxiliary controller to illustrate feasibility and stability\cite{Wang2020}\cite{Mhaskar2006}. We extend it to a stochastic framework, namely,
\begin{eqnarray}
\mathbf{\phi}\left( \mathbf{x} \right) = \left\{ \begin{matrix}
{- \frac{a + \sqrt{a^{2} + \left\| {u^{max}\mathbf{b}} \right\|^{4}}}{\left\| \mathbf{b} \right\|^{2}\left\lbrack {1 + \sqrt{1 + \left\| {u^{max}\mathbf{b}} \right\|^{2}}} \right\rbrack}\mathbf{b}^{T},\mathbf{b} \neq 0} \\
{0,\mathbf{b} = 0} \\
\end{matrix} \right.
\end{eqnarray}
where $a = L_{f}W_{c}\left( \mathbf{x} \right) + \frac{1}{2}\mathrm{tr}\left( {\mathbf{\sigma}^{T}\left( \mathbf{x} \right)\frac{\partial^{2}W_{c}}{\partial\mathbf{x}^{2}}\mathbf{\sigma}\left( \mathbf{x} \right)} \right) + \rho V_{c}\left( \mathbf{x} \right)$, $\mathbf{b} = L_{g}W_{c}\left( \mathbf{x} \right)$, $\left\| \mathbf{u} \right\| \leq u^{max}$, $\rho > 0$.
\begin{proposition}\label{pro4}
Let $X_{\phi} = \left\{ {\left. \mathbf{x} \right|a \leq \left\| {u^{max}\mathbf{b}} \right\|} \right\}$. For $\forall\mathbf{x} \in X_{\phi},\mathbf{x} \neq 0$, $\mathbf{\phi}\left( \mathbf{x} \right) \in U$ and $\mathcal{L}W_{c}\left( {\mathbf{x},\mathbf{\phi}\left( \mathbf{x} \right)} \right) \leq - \rho V_{c}\left( \mathbf{x} \right) < 0$.
\end{proposition}
\begin{proof}
The conclusion is true if $\mathbf{b} = 0$. We now prove it is true when $\mathbf{b} \neq 0$.
\begin{enumerate}[a.]
\item $\mathbf{\phi}\left( \mathbf{x} \right) \in U$
\begin{enumerate}[1)]
\item $0 \leq a \leq \left\| {u^{max}\mathbf{b}} \right\|$
\begin{eqnarray*}
\left| {a + \sqrt{a^{2} + \left\| {u^{max}\mathbf{b}} \right\|^{4}}} \right| \leq |a| + \left| \sqrt{a^{2} + \left\| {u^{max}\mathbf{b}} \right\|^{4}} \right| \leq \left\| {u^{max}\mathbf{b}} \right\| + \sqrt{\left\| {u^{max}\mathbf{b}} \right\|^{2} + \left\| {u^{max}\mathbf{b}} \right\|^{4}}
\end{eqnarray*}
It follows that
\begin{eqnarray*}
\left| {- \frac{a + \sqrt{a^{2} + \left\| {u^{max}\mathbf{b}} \right\|^{4}}}{\left\| \mathbf{b} \right\|^{2}\left\lbrack 1 + \sqrt{1 + \left\| {u^{max}\mathbf{b}} \right\|^{2}} \right\rbrack}} \right| \leq \frac{\left\| {u^{max}\mathbf{b}} \right\| + \sqrt{\left\| {u^{max}\mathbf{b}} \right\|^{2} + \left\| {u^{max}\mathbf{b}} \right\|^{4}}}{\left\| \mathbf{b} \right\|^{2}\left\lbrack 1 + \sqrt{1 + \left\| {u^{max}\mathbf{b}} \right\|^{2}} \right\rbrack} = \frac{\left\| {u^{max}\mathbf{b}} \right\|}{\left\| \mathbf{b} \right\|^{2}}
\end{eqnarray*}
We have
\begin{eqnarray*}
\left\| {\mathbf{\phi}\left( \mathbf{x} \right)} \right\| \leq \frac{\left\| {u^{max}\mathbf{b}} \right\|}{\left\| \mathbf{b} \right\|^{2}}\left\| \mathbf{b}^{T} \right\| = u^{max}
\end{eqnarray*}
which means
\begin{eqnarray*}
\mathbf{\phi}\left( \mathbf{x} \right) \in U
\end{eqnarray*}
\item $a < 0 \leq \left\| {u^{max}\mathbf{b}} \right\|$
\begin{eqnarray*}
a^{2} + \left\| {u^{max}\mathbf{b}} \right\|^{4} \leq \left( {\left\| {u^{max}\mathbf{b}} \right\|^{2} - a} \right)^{2} = a^{2} + \left\| {u^{max}\mathbf{b}} \right\|^{4} - 2a\left\| {u^{max}\mathbf{b}} \right\|^{2}
\end{eqnarray*}
thus
\begin{eqnarray*}
a + \sqrt{a^{2} + \left\| {u^{max}\mathbf{b}} \right\|^{4}} \leq \left\| {u^{max}\mathbf{b}} \right\|^{2}
\end{eqnarray*}
It follows that
\begin{eqnarray*}
\left| {- \frac{a + \sqrt{a^{2} + \left\| {u^{max}\mathbf{b}} \right\|^{4}}}{\left\| \mathbf{b} \right\|^{2}\left\lbrack 1 + \sqrt{1 + \left\| {u^{max}\mathbf{b}} \right\|^{2}} \right\rbrack}} \right| = \frac{a + \sqrt{a^{2} + \left\| {u^{max}\mathbf{b}} \right\|^{4}}}{\left\| \mathbf{b} \right\|^{2}\left\lbrack 1 + \sqrt{1 + \left\| {u^{max}\mathbf{b}} \right\|^{2}} \right\rbrack} \leq \frac{\left\| {u^{max}\mathbf{b}} \right\|^{2}}{\left\| \mathbf{b} \right\|^{2}\left\lbrack 1 + \sqrt{1 + \left\| {u^{max}\mathbf{b}} \right\|^{2}} \right\rbrack} \leq \frac{\left\| {u^{max}\mathbf{b}} \right\|}{\left\| \mathbf{b} \right\|^{2}}
\end{eqnarray*}
We have
\begin{eqnarray*}
\left\| {\mathbf{\phi}\left( \mathbf{x} \right)} \right\| \leq \frac{\left\| {u^{max}\mathbf{b}} \right\|}{\left\| \mathbf{b} \right\|^{2}}\left\| \mathbf{b}^{T} \right\| = u^{max}
\end{eqnarray*}
which means
\begin{eqnarray*}
\mathbf{\phi}\left( \mathbf{x} \right) \in U
\end{eqnarray*}
\end{enumerate}
\item $LW_{c}\left( {\mathbf{x},\mathbf{\phi}\left( \mathbf{x} \right)} \right) < 0$
\begin{enumerate}[1)]
\item $0 \leq a \leq \left\| {u^{max}\mathbf{b}} \right\|$
\begin{align*}
\mathcal{L}W_{c}\left( {\mathbf{x},\mathbf{\phi}\left( \mathbf{x} \right)} \right) & = L_{f}W_{c}\left( \mathbf{x} \right) + \frac{1}{2}\mathrm{tr}\left( {\mathbf{\sigma}^{T}\left( \mathbf{x} \right)\frac{\partial^{2}W_{c}}{\partial\mathbf{x}^{2}}\mathbf{\sigma}\left( \mathbf{x} \right)} \right) + L_{g}W_{c}\left( \mathbf{x} \right)\mathbf{\phi}\left( \mathbf{x} \right) = a - \rho V_{c}\left( \mathbf{x} \right) + \mathbf{b}\mathbf{\phi}\left( \mathbf{x} \right) \\
& = - \rho V_{c}\left( \mathbf{x} \right) + \frac{a\left\lbrack {1 + \sqrt{1 + \left\| {u^{max}\mathbf{b}} \right\|^{2}}} \right\rbrack - a - \sqrt{a^{2} + \left\| {u^{max}\mathbf{b}} \right\|^{4}}}{\left\lbrack {1 + \sqrt{1 + \left\| {u^{max}\mathbf{b}} \right\|^{2}}} \right\rbrack} \\
& = - \rho V_{c}\left( \mathbf{x} \right) + \frac{\sqrt{a^{2} + \left\| {au^{max}\mathbf{b}} \right\|^{2}} - \sqrt{a^{2} + \left\| {u^{max}\mathbf{b}} \right\|^{4}}}{\left\lbrack {1 + \sqrt{1 + \left\| {u^{max}\mathbf{b}} \right\|^{2}}} \right\rbrack} \leq - \rho V_{c}\left( \mathbf{x} \right) < 0
\end{align*}
\item $a < 0 \leq \left\| {u^{max}\mathbf{b}} \right\|$
\begin{align*}
\mathcal{L}W_{c}\left( {\mathbf{x},\mathbf{\phi}\left( \mathbf{x} \right)} \right)  = - \rho V_{c}\left( \mathbf{x} \right) + \frac{a\sqrt{1 + \left\| {u^{max}\mathbf{b}} \right\|^{2}} - \sqrt{a^{2} + \left\| {u^{max}\mathbf{b}} \right\|^{4}}}{\left\lbrack {1 + \sqrt{1 + \left\| {u^{max}\mathbf{b}} \right\|^{2}}} \right\rbrack} < - \rho V_{c}\left( \mathbf{x} \right) < 0
\end{align*}
\end{enumerate}
\end{enumerate}
Therefore, $\mathbf{\phi}\left( \mathbf{x} \right) \in U$ and $\mathcal{L}W_{c}\left( {\mathbf{x},\mathbf{\phi}\left( \mathbf{x} \right)} \right) \leq - \rho V_{c}\left( \mathbf{x} \right) < 0,\forall\mathbf{x} \in X_{\phi},\mathbf{x} \neq 0$.
\end{proof}
Input constraints in Proposition \ref{pro4} take the form of $\left\| \mathbf{u} \right\| \leq u^{max}$ and includes zero input. We now consider the form $\mathbf{u}_{\mathbf{m}\mathbf{i}\mathbf{n}} \leq \mathbf{u} \leq \mathbf{u}_{\mathbf{m}\mathbf{a}\mathbf{x}}$. Let $\mathbf{u} = \mathbf{u}_{\mathbf{m}\mathbf{e}\mathbf{a}\mathbf{n}} + \mathbf{u}_{\mathbf{d}}\mathbf{K}$, where $\mathbf{u}_{\mathbf{m}\mathbf{e}\mathbf{a}\mathbf{n}} = \left( {\frac{u_{max1} + u_{min1}}{2},\ldots,\frac{u_{maxm} + u_{minm}}{2}} \right)$, $
\mathbf{u}_{\mathbf{d}} = diag\left( {\frac{u_{max1} - u_{min1}}{2},\ldots,\frac{u_{maxm} - u_{minm}}{2}} \right)$ and $\mathbf{b} = L_{g}W_{c}\left( \mathbf{x} \right)\mathbf{u}_{\mathbf{d}}$. 
Substitute $a = L_{f}W_{c}\left( \mathbf{x} \right) + \frac{1}{2}\mathrm{tr}\left( {\mathbf{\sigma}^{T}\left( \mathbf{x} \right)\frac{\partial^{2}W_{c}}{\partial\mathbf{x}^{2}}\mathbf{\sigma}\left( \mathbf{x} \right)} \right) + L_{g}W_{c}\left( \mathbf{x} \right)\mathbf{u}_{\mathbf{m}\mathbf{e}\mathbf{a}\mathbf{n}} + \rho V_{c}\left( \mathbf{x} \right)$ and $\mathbf{b} = L_{g}W_{c}\left( \mathbf{x} \right)\mathbf{u}_{\mathbf{d}}$ into Proposition \ref{pro4}, then following bounded control law satisfies $\mathbf{\phi}\left( \mathbf{x} \right) \in U$ and $\mathcal{L}W_{c}\left( {\mathbf{x},\mathbf{\phi}\left( \mathbf{x} \right)} \right) < 0$, $\forall\mathbf{x} \in X_{\phi} = \left\{ {\left. \mathbf{x} \right|a \leq \left\| \mathbf{b} \right\|} \right\},\mathbf{x} \neq 0$.
\begin{eqnarray*}
\mathbf{\phi}\left( \mathbf{x} \right) = \left\{ \begin{matrix}
{\mathbf{u}_{\mathbf{m}\mathbf{e}\mathbf{a}\mathbf{n}} - \frac{a + \sqrt{a^{2} + \left\| \mathbf{b} \right\|^{4}}}{\left\| \mathbf{b} \right\|^{2}\left\lbrack {1 + \sqrt{1 + \left\| \mathbf{b} \right\|^{2}}} \right\rbrack}\mathbf{u}_{\mathbf{d}}\mathbf{b}^{T},b \neq 0} \\
{\mathbf{u}_{\mathbf{m}\mathbf{e}\mathbf{a}\mathbf{n}},b = 0} \\
\end{matrix} \right.
\end{eqnarray*}
To enlarge feasibility region and reduce constraint violations near boundary of unsafe region, we should illustrate approximation of feasibility region in control design. With $X$ in (\ref{eq29}) being replaced by $X_{\phi} = \left\{ {\left. \mathbf{x} \right|a \leq \left\| \mathbf{b} \right\|} \right\}$, the stochastic model predictive control problem is feasible $\forall\mathbf{x} \in X_{\phi}\backslash D$. The feasible region $X_\phi$ is related to a specific candidate control if we take $\mathbf{\phi}\left( \mathbf{x}\right)$ as an auxiliary controller. However, for MPC, such candidate control is not necessary, and a larger feasible region can be pursued\cite{Mahmood2008}, i.e., $X_\phi$ can be extended to the region that negative definiteness of the Lyapunov function derivative can be achieved. For stochastic system, we define
\begin{eqnarray*}
X_{L} = \left\{ {\left. \mathbf{x} \right|L_{f}W_{c}\left( \mathbf{x} \right) + {\sum\limits_{i = 1}^{n}{L_{G}W_{c}\left( \mathbf{x} \right)u^{i}}} + \frac{1}{2} \mathrm{tr}\left( {\sigma^{T}\left( \mathbf{x} \right)\frac{\partial^{2}W_{c}}{\partial\mathbf{x}^{2}}\sigma\left( \mathbf{x} \right)} \right) < 0} \right\}
\end{eqnarray*}
where $L_{G}W_{c}\left( \mathbf{x} \right)u^{i} = L_{g}W_{c}\left( \mathbf{x} \right)u_{maxi}$ if $L_{g}W_{c}\left( \mathbf{x} \right) \leq 0$ and $L_{G}W_{c}\left( \mathbf{x} \right)u^{i} = L_{g}W_{c}\left( \mathbf{x} \right)u_{mini}$ if $L_{g}W_{c}\left( \mathbf{x} \right) > 0$. 
$X_{L}$ only combines constraints $\mathbf{u}\left( i \middle| k \right) \in U$ and $\mathcal{L}W_{c} < 0$ at first step in predictive horizon and therefore is an external approximation of feasibility region. If weighting matrix $R_2$ is large enough, solution of proposed CLBF-based Stochastic MPC is equivalent to original optimization ($\delta(0|k)=0$) as long as feasible solution exists and optimization is always feasible $\forall\mathbf{x} \in X_{L}\backslash D$. Note that the invariance of $X\backslash D_{relaxed}$ in Proposition 1 is conditional on the fact that control action $\mathbf{u}_{t}$ satisfies $\mathcal{L}W_{c}\left( \mathbf{x}_{t} \right) < 0$ all the time. However, unboundedness of disturbances and sample-and-hold implementation of MPC may invalid such condition with a low probability\cite{Mahmood2012,Wu2018,Homer2017}.
Therefore we refer to
stability analysis for stochastic nonlinear sampled-data systems\cite{gao2016estimation,luo2019event} to analyse performance during sampling intervals. We give the following lemma and proposition.

\begin{lemma}\label{lem1}\cite{luo2019event}
Consider sampled-data system
\begin{gather}\label{eq31}
d\mathbf{x}_{t} = \left( {f\left( \mathbf{x}_{t} \right) + g\left( \mathbf{x}_{t} \right)k\left( {\mathbf{x}_{t} + \mathbf{e}_{t}} \right)} \right)dt + \sigma\left( \mathbf{x}_{t} \right)dW
\\
\mathbf{e}_{t} = \mathbf{x}_{t_{k}} - \mathbf{x}_{t},t \in \left\lbrack {t_{k},t_{k + 1}} \right)
\end{gather}
where the control input with
sample-and-hold state measurements has the form $u(t)=k\left( \mathbf{x}_{t_k} \right), t \in \left\lbrack {t_{k},t_{k + 1}} \right). f\left( \mathbf{x} \right),\sigma\left( \mathbf{x} \right),k\left( \mathbf{x} \right)$ are locally Lipschitz continuous, and there exist positive scalars $L_1,L_2,L_3,L_4$ such that $\forall\mathbf{x} \in X, \|f\left( \mathbf{x} \right)\| \leq L_1\|\mathbf{x}\|, \|\sigma\left( \mathbf{x} \right)\| \leq L_2\|\mathbf{x}\|, \|k\left( \mathbf{x} \right)\| \leq L_3\|\mathbf{x}\|, \|g\left( \mathbf{x} \right)\| \leq L_4$.
Suppose there exists a positive definite function $
V\left( {\mathbf{x},t} \right) \in C^{2,1}$, such that there are positive constants $c_1^{'},c_2^{'},c_3^{'},c_4^{'}, \forall{\left({t,\mathbf{x},\mathbf{e}}\right)} \in \left\lbrack {t_{0},\infty} \right) \times X \times X$ satisfying
\begin{gather}\label{eq33}
 c_1^{'}\left\| \mathbf{x} \right\|^{2}  \leq V\left( {\mathbf{x},t} \right) \leq c_2^{'}\left\| \mathbf{x} \right\|^{2} 
\\\label{eq34}
\mathcal{L}V\left( {\mathbf{x},t} \right) \leq c_3^{'}\left\| \mathbf{e} \right\|^{2}  - c_4^{'}\left\| \mathbf{x} \right\|^{2} 
\end{gather}
If sampling period $T \in \left\lbrack {0,T^{*}} \right)$, where $T^{*}$ is the unique solution of
\begin{equation}
c_4^{'} \varphi \left( t\right)= 2c_3^{'} \rho
\left( t\right)
\end{equation}
and $\varphi \left( t\right) = e^{-\alpha t}+ \frac{L_3^2}{\alpha}\left( e^{-\alpha t} -1\right),\rho
\left( t\right)=4t\left( 2tL_1^2+L_2^2+tL_3^2L_4^2\right)e^{4t\left( 2tL_1^2+L_2^2\right)},\alpha=2L_1+L_2^2+L_4^2$, the closed-loop system is mean-square exponentially stable.
\end{lemma}

\begin{proposition}\label{pro5}
For sampled-data system (\ref{eq31}), under the proposed CLBF based stochastic MPC, if auxiliary Lyapunov controller satisfying
\begin{equation}\label{eq36}
\mathcal{L}V_{c}\left( {\mathbf{x}_{t},\mathbf{\phi}\left( \mathbf{x}_{t} \right)} \right) \leq - \mu_{1}\left\| \mathbf{x}_{t} \right\|^{2} - \mu_{2}\left\| \frac{\partial V_{c}}{\partial\mathbf{x}_{t}} \right\|^{2}
\end{equation}
\begin{equation}\label{eq37}
\left\| {L_{g}V_{c}\left( \mathbf{x}_{t} \right) - L_{g}V_{c}\left( \mathbf{x}_{t_{k}} \right)} \right\| \leq \mu_{3}\left\| \mathbf{x}_{t} \right\|^{2} + \mu_{4}\left\| \mathbf{e}_{t} \right\|^{2}
\end{equation}
and if $\mu_{1} - 2u^{max}\mu_{3}>0$, where $\left\| \mathbf{u} \right\| \leq u^{max}$, the closed-loop system is
mean-square exponentially stable on condition that sampling period $T \in \left\lbrack {0,T^{*}} \right)$ in Lemma \ref{lem1}.
\end{proposition}

\begin{proof}
For $0\notin X_{i}$, let $
c_{5i} = \min\limits_{\mathbf{x} \in X_{i}}\left\| \mathbf{x} \right\|^{2}>0$, we have
\begin{align*}
c_{1}\left\| \mathbf{x} \right\|^{2} & \leq V_{c}\left( \mathbf{x} \right) \leq c_{2}\left\| \mathbf{x} \right\|^{2} + {\sum\limits_{i = 1}^{n_{B}}{\lambda_{i}d_{Xi}\left( {B_{imax} - B_{imin}} \right)}}  = c_{2}\left\| \mathbf{x} \right\|^{2} + {\sum\limits_{i = 1}^{n_{B}}{\lambda_{i}d_{Xi}\frac{B_{imax} - B_{imin}}{c_{5i}}\begin{matrix}
{\min} \\
{\mathbf{x} \in X_{i}} \\
\end{matrix}\left\| \mathbf{x} \right\|^{2}}}  \\
& \leq c_{2}\left\| \mathbf{x} \right\|^{2} + {\sum\limits_{i = 1}^{n_{B}}{\lambda_{i}d_{Xi}\frac{B_{imax} - B_{imin}}{c_{5i}}\left\| \mathbf{x} \right\|^{2}}}  \leq \left\lbrack {c_{2} + {\max\left( {\lambda_{i}\frac{B_{imax} - B_{imin}}{c_{5i}}} \right)}} \right\rbrack\left\| \mathbf{x} \right\|^{2}
\end{align*}
Let $c_1^{'}=c_1,c_2^{'}=c_{2} + {\max\left( {\lambda_{i}\frac{B_{imax} - B_{imin}}{c_{5i}}} \right)}$, we obtain (\ref{eq33}).
According to (\ref{eq36}),
\begin{align*}
\mathcal{L}V_{c}\left( {\mathbf{x}_{t},\mathbf{\phi}\left( {\mathbf{x}_{t} + \mathbf{e}_{t}} \right)} \right) 
& \leq - \mu_{1}\left\| \mathbf{x}_{t} \right\|^{2} - \mu_{2}\left\| \frac{\partial V_{c}}{\partial\mathbf{x}_{t}} \right\|^{2} + \frac{\partial V_{c}}{\partial\mathbf{x}_{t}}\mathbf{g}\left( \mathbf{x}_{t} \right)\left( {\mathbf{\phi}\left( {\mathbf{x}_{t} + \mathbf{e}_{t}} \right) - \mathbf{\phi}\left( \mathbf{x}_{t} \right)} \right) \\
& \leq - \mu_{1}\left\| \mathbf{x}_{t} \right\|^{2} + \frac{1}{4\mu_{2}}\left\| {\mathbf{g}\left( \mathbf{x}_{t} \right)\left( {\mathbf{\phi}\left( {\mathbf{x}_{t} + \mathbf{e}_{t}} \right) - \mathbf{\phi}\left( \mathbf{x}_{t} \right)} \right)} \right\|^{2} \leq - \mu_{1}\left\| \mathbf{x}_{t} \right\|^{2} + \frac{L_{3}^{2}L_{4}^{2}}{4\mu_{2}}\left\| \mathbf{e}_{t} \right\|^{2}
\end{align*}
Moreover,
\begin{gather*}
\mathcal{L}V_{c}\left( {\mathbf{x}_{t},\mathbf{u}_{t_{k}}} \right) - \mathcal{L}V_{c}\left( {\mathbf{x}_{t},\mathbf{\phi}\left( \mathbf{x}_{t_{k}} \right)} \right) = L_{g}V_{c}\left( \mathbf{x}_{t} \right)\left( {\mathbf{u}_{t_{k}} - \mathbf{\phi}\left( \mathbf{x}_{t_{k}} \right)} \right)
\\
\mathcal{L}V_{c}\left( {\mathbf{x}_{t_{k}},\mathbf{u}_{t_{k}}} \right) - \mathcal{L}V_{c}\left( {\mathbf{x}_{t_{k}},\mathbf{\phi}\left( \mathbf{x}_{t_{k}} \right)} \right) = L_{g}V_{c}\left( \mathbf{x}_{t_{k}} \right)\left( {\mathbf{u}_{t_{k}} - \mathbf{\phi}\left( \mathbf{x}_{t_{k}} \right)} \right) \leq 0
\end{gather*}
Substituting (\ref{eq37}) in it, we have
\begin{align*}
\mathcal{L}V_{c}\left( {\mathbf{x}_{t},\mathbf{u}_{t_{k}}} \right) 
& = \left\lbrack {\mathcal{L}V_{c}\left( {\mathbf{x}_{t},\mathbf{u}_{t_{k}}} \right) - \mathcal{L}V_{c}\left( {\mathbf{x}_{t_{k}},\mathbf{u}_{t_{k}}} \right)} \right\rbrack - \left\lbrack {\mathcal{L}V_{c}\left( {\mathbf{x}_{t},\mathbf{\phi}\left( \mathbf{x}_{t_{k}} \right)} \right) - \mathcal{L}V_{c}\left( {\mathbf{x}_{t_{k}},\mathbf{\phi}\left( \mathbf{x}_{t_{k}} \right)} \right)} \right\rbrack \\
& + \left\lbrack {\mathcal{L}V_{c}\left( {\mathbf{x}_{t_{k}},\mathbf{u}_{t_{k}}} \right) - \mathcal{L}V_{c}\left( {\mathbf{x}_{t_{k}},\mathbf{\phi}\left( \mathbf{x}_{t_{k}} \right)} \right)} \right\rbrack + \mathcal{L}V_{c}\left( {\mathbf{x}_{t},\mathbf{\phi}\left( \mathbf{x}_{t_{k}} \right)} \right)
\\
&  \leq \left\lbrack {\mathcal{L}V_{c}\left( {\mathbf{x}_{t},\mathbf{u}_{t_{k}}} \right) - \mathcal{L}V_{c}\left( {\mathbf{x}_{t_{k}},\mathbf{u}_{t_{k}}} \right)} \right\rbrack - \left\lbrack {\mathcal{L}V_{c}\left( {\mathbf{x}_{t},\mathbf{\phi}\left( \mathbf{x}_{t_{k}} \right)} \right) - \mathcal{L}V_{c}\left( {\mathbf{x}_{t_{k}},\mathbf{\phi}\left( \mathbf{x}_{t_{k}} \right)} \right)} \right\rbrack - \mu_{1}\left\| \mathbf{x}_{t} \right\|^{2} + \frac{L_{3}^{2}L_{4}^{2}}{4\mu_{2}}\left\| \mathbf{e}_{t} \right\|^{2} \\
& = \left\lbrack {L_{g}V_{c}\left( \mathbf{x}_{t} \right) - L_{g}V_{c}\left( \mathbf{x}_{t_{k}} \right)} \right\rbrack\left\lbrack {\mathbf{u}_{t_{k}} - \mathbf{\phi}\left( \mathbf{x}_{t_{k}} \right)} \right\rbrack - \mu_{1}\left\| \mathbf{x}_{t} \right\|^{2} + \frac{L_{3}^{2}L_{4}^{2}}{4\mu_{2}}\left\| \mathbf{e}_{t} \right\|^{2} \\
& \leq 2u^{max}\left\| {L_{g}V_{c}\left( \mathbf{x}_{t} \right) - L_{g}V_{c}\left( \mathbf{x}_{t_{k}} \right)} \right\| - \mu_{1}\left\| \mathbf{x}_{t} \right\|^{2} + \frac{L_{3}^{2}L_{4}^{2}}{4\mu_{2}}\left\| \mathbf{e}_{t} \right\|^{2} \\
& \leq - \left( {\mu_{1} - 2u^{max}\mu_{3}} \right)\left\| \mathbf{x}_{t} \right\|^{2} + \left( {\frac{L_{3}^{2}L_{4}^{2}}{4\mu_{2}} + 2u^{max}\mu_{4}} \right)\left\| \mathbf{e}_{t} \right\|^{2}
\end{align*}
and obtain (\ref{eq34}). By Lemma \ref{lem1}, the closed-loop system is mean-square exponentially stable.
\end{proof}
Note that in the proof we assume $\delta\left( 0 \middle| k \right) = 0$ (i.e. $\forall\mathbf{x} \in X_{\phi}\backslash D$), if $\delta\left( 0 \middle| k \right) > 0$ (i.e. $\forall\mathbf{x} \in X_{L}\backslash X_{\phi}$) to recover feasibility, the conclusion can be relaxed to ultimately bounded in the mean square.
To enlarge maximum allowable sampling period  and enhance control performance, event-triggering mechanisms can be integrated into CLBF based stochastic MPC, as illustrated by Figure \ref{fig6}. When the event is triggered at $t_k$, the control input $\mathbf{u}_{t_k}$ with sample-and-hold state measurements $\mathbf{x}_{t_k}$ is obtained through the proposed stochastic MPC and kept until the next trigger. We summarize the algorithm in Algorithm \ref{tab0}.

\begin{figure}[htbp]
\centerline{\includegraphics[width=200pt,height=11pc]{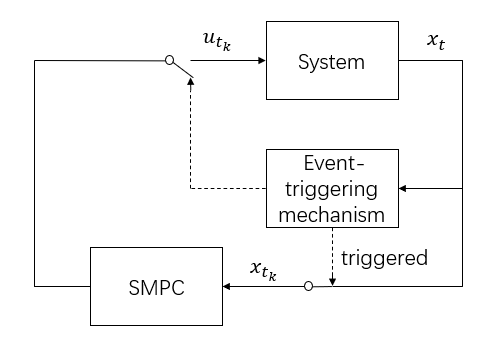}}
\caption{Control structure of sampled-data SMPC with  event-triggering mechanisms\label{fig6}}
\end{figure}

\begin{algorithm}[htbp]
\caption{CLBF based stochastic MPC algorithm}\label{tab0}
1. Measure the initial states $\mathbf{x}_{t_0}$.\\
2. Sample the current states $\mathbf{x}\left( k \right)$ at $t=t_k$ and initialize $\mathbf{x}\left( 0|k \right)=\mathbf{x}\left( k \right)$. \\
3. Discretize nominal dynamics and compute predicted states $\mathbf{x}\left( i \middle| k \right)$ for $i=1,2,\cdots,N$.\\
4. Solve the stochastic model predictive control problem and obtain predicted input $\mathbf{u}\left( i \middle| k \right)$ for $i=0,1,\cdots,N-1$.\\
5. while measuring the current states $\mathbf{x}_t$, implement and hold the control input $\mathbf{u}_t=\mathbf{u}\left( 0 \middle| k \right)$ until the event is triggered.\\ 
6. If the event is triggered at $t=t_{k+1}$, let $k \gets k+1$, turn to 2.\\
\end{algorithm}

\begin{proposition}\label{pro6}
On the basis of MPC controller in Proposition \ref{pro5}, the closed-loop system is mean-square exponentially stable if triggered time instant $t_{k + 1}$ is determined by the following triggering condition
\begin{equation}\label{eq38}
t_{k + 1} = {\inf\left\{ t \geq t_{k} + \tau , t \leq t_{k} + \tau_{max} \middle| \mathcal{L}V_{c}\left( {\mathbf{x}_{t},\mathbf{u}_{t_{k}}} \right) \leq \mathcal{L}V_{c}\left( {\mathbf{x}_{t},\mathbf{\phi}\left( \mathbf{x}_{t} \right)} \right) \right\}}
\end{equation}
where $\tau \leq T^{*}$ is the minimum interexecution time to avoid Zeno phenomena, i.e., infinite triggers in finite
time.
\end{proposition}

\begin{proof}
As discussed in Remark 2\cite{luo2019event}, if $\tau \leq T^{*}$, for $t \in \left(t_{k},t_{k} + \tau\right)$,
\begin{equation*}
E\mathcal{L}V_{c}\left( {\mathbf{x}_{t},\mathbf{u}_{t_{k}}} \right) \leq  - \frac{1}{2}c_4^{'} E\left\| \mathbf{x}_{t} \right\|^{2}
\end{equation*}
Moreover, by triggering condition (\ref{eq38}), for $t \in \left[t_{k} + \tau,t_{k+1}\right]$,  
\begin{equation*}
\mathcal{L}V_{c}\left( {\mathbf{x}_{t},\mathbf{u}_{t_{k}}} \right) \leq \mathcal{L}V_{c}\left( {\mathbf{x}_{t},\mathbf{\phi}\left( \mathbf{x}_{t} \right)} \right) \leq - \rho c_1\left\| \mathbf{x}_{t} \right\|^{2}
\end{equation*}
Hence, for $t \in \left(t_{k},t_{k+1}\right]$,
\begin{equation*}
E\mathcal{L}V_{c}\left( {\mathbf{x}_{t},\mathbf{u}_{t_{k}}} \right) \leq  -c_0 E\left\| \mathbf{x}_{t} \right\|^{2}
\end{equation*}
where $c_0=\min\left(\frac{1}{2}c_4^{'},\rho c_1\right)$, and the closed-loop system is mean-square exponentially stable according to Lemma 2\cite{gao2016estimation}.
\end{proof}

\begin{remark}\label{rem1}
Selection of parameters mainly influences feasibility region, stability, and convergence rate, therefore there is a trade-off between safety and stability. We should first decide $B_i\left( \mathbf{x} \right)$ to maximize feasibility region, and then tune $V\left( \mathbf{x} \right)$ and $\rho$ to decide convergence rate and to obtain a proper $T^{*}$. We will give an example to show design procedure in Section \ref{sec5}. In Figure \ref{fig3}, we first consider feasibility region only and obtain a proper $B_i\left( \mathbf{x} \right)$, which makes $X_L$ contain $X\backslash D$. Then in the right graph, we change parameters in $V\left( \mathbf{x} \right)$  to obtain a proper $T^{*}$ that reaches stability condition for sampled-data systems.
\end{remark}

\section{Application to Wheeled Mobile Robots}\label{sec5}
We consider a wheeled mobile robot example. Our goal is to drive the robot to the origin $\left( {x,y} \right) = \left( {0,0} \right)$ while avoiding the unsafe sets $
D_{i} = \left\{ \mathbf{x} \middle| F_{Bi}\left( \mathbf{x} \right) = \left( {x - x_{i\_ obs}~} \right)^{2} + \left( {y - y_{i\_ obs}} \right)^{2} < l_{Di} \right\}$. The corresponding $X_i$ are chosen as $X_{i} = \left\{ \mathbf{x} \middle| F_{Bi}\left( \mathbf{x} \right) < l_{Xi} \right\}$. Kinematic model of the wheeled mobile robot is described by
\begin{eqnarray*}
d\mathbf{x} = g\left( \mathbf{x} \right)\mathbf{u}dt + \sigma\left( \mathbf{x} \right)dW
\end{eqnarray*}
with
\begin{eqnarray*}
g\left( \mathbf{x} \right) = \begin{bmatrix}
{\cos\theta} & 0 \\
{\sin\theta} & 0 \\
0 & 1 \\
\end{bmatrix},
\sigma\left( \mathbf{x} \right) = diag\left( {\left( {1 - d_{Xg} + d_{Xg}\frac{\sqrt{x^{2} + y^{2}}}{r_{g}}} \right)\sigma_{1},\left( {1 - d_{Xg} + d_{Xg}\frac{\sqrt{x^{2} + y^{2}}}{r_{g}}} \right)\sigma_{2},\sigma_{3}} \right)
\end{eqnarray*}
\begin{eqnarray*}
d_{Xg} = \left\{ \begin{matrix}
{1,\forall x \in X_{g}} \\
{0,\forall x \in X\backslash X_{g}} \\
\end{matrix} \right.,X_{g} = \left\{ \mathbf{x} \middle| {x^{2} + y^{2} \leq r_{g}^{2}} \right\}
\end{eqnarray*}
where $\mathbf{x} = \begin{bmatrix}
x & y & \theta \\
\end{bmatrix}^{T}$ denotes the position and orientation of the mobile robot; $\mathbf{u} = \begin{bmatrix}
v & \omega \\
\end{bmatrix}^{T}$ denotes the translational velocity and angular velocity; $W$ is the vector of independent Brownian motions.

Constraints on control inputs are given by $
\begin{bmatrix}
{- 10} \\
{- \frac{\pi}{2}} \\
\end{bmatrix} \leq \mathbf{u} \leq \begin{bmatrix}
10 \\
\frac{\pi}{2} \\
\end{bmatrix}
$. Let $\bar{\mathbf{x}} = \begin{bmatrix}
x & y & \theta & v \\
\end{bmatrix}^{T}$, $
\mathbf{z} = \mathbf{\phi}\left( \bar{\mathbf{x}} \right) = \begin{bmatrix}
x & y & {v\cos\theta} & {v\sin\theta} \\
\end{bmatrix}^{T}$. The semi-linear system in (\ref{eq16}) is
\begin{eqnarray*}
d\mathbf{z} = \begin{bmatrix}
0 & 0 & 1 & 0 \\
0 & 0 & 0 & 1 \\
0 & 0 & {- \frac{1}{2}\sigma_{3}^{2}} & 0 \\
0 & 0 & 0 & {- \frac{1}{2}\sigma_{3}^{2}} \\
\end{bmatrix}\mathbf{z}dt + \begin{bmatrix}
0 & 0 \\
0 & 0 \\
1 & 0 \\
0 & 1 \\
\end{bmatrix}\mathbf{u}^{\prime}dt + \mathbf{\sigma}^{\prime}\left( \mathbf{z} \right)dW
\end{eqnarray*}
with
\begin{eqnarray*}
\mathbf{u}^{\prime} = \begin{bmatrix}
{- v\sin\theta} & {\cos\theta} \\
{v\cos\theta} & {\sin\theta} \\
\end{bmatrix}\begin{bmatrix}
\omega \\
\dot{v} \\
\end{bmatrix},
\mathbf{\sigma}^{\prime}\left( \mathbf{z} \right) = \begin{bmatrix}
{\left( {1 - d_{Xg} + d_{Xg}\frac{\sqrt{z_{1}^{2} + z_{2}^{2}}}{r_{g}}} \right)\sigma_{1}} & 0 & 0 \\
0 & {\left( {1 - d_{Xg} + d_{Xg}\frac{\sqrt{z_{1}^{2} + z_{2}^{2}}}{r_{g}}} \right)\sigma_{2}} & 0 \\
0 & 0 & {- \sigma_{3}z_{4}} \\
0 & 0 & {\sigma_{3}z_{3}} \\
\end{bmatrix}
\end{eqnarray*}
We choose a quadratic function $\bar{V}\left( \mathbf{z} \right) = \mathbf{z}^{T}\mathbf{P}\mathbf{z}$ where
\begin{eqnarray*}
\mathbf{P} = \begin{bmatrix}
p_{1} & 0 & p_{2} & 0 \\
0 & p_{1} & 0 & p_{2} \\
p_{2} & 0 & p_{3} & 0 \\
0 & p_{2} & 0 & p_{3} \\
\end{bmatrix}
\end{eqnarray*}
with $p_1>0,p_3>0,p_1p_3-p_2^2>0$. The infinitesimal generator of $\bar{V}\left( \mathbf{z} \right)$ is
\begin{align*}
\mathcal{L}\bar{V}\left( \mathbf{z} \right) & = 2\left( {p_{3}z_{3} + p_{2}z_{1}} \right)u_{1}^{\prime} + 2\left( {p_{3}z_{4} + p_{2}z_{2}} \right)u_{2}^{\prime} + \left( {2p_{1} - p_{2}\sigma_{3}^{2}} \right)\left( {z_{1}z_{3} + z_{2}z_{4}} \right) + 2p_{2}\left( {z_{3}^{2} + z_{4}^{2}} \right) \\
& + p_{1}\left( {\sigma_{1}^{2} + \sigma_{2}^{2}} \right)\left( {1 - d_{Xg} + d_{Xg}\frac{z_{1}^{2} + z_{2}^{2}}{r_{g}^{2}}} \right)
\end{align*}
When $L_{B}\bar{V}\left( \mathbf{z} \right) = 2\begin{bmatrix}
{p_{3}z_{3} + p_{2}z_{1}} & {p_{3}z_{4} + p_{2}z_{2}} \\
\end{bmatrix} = 0$, 
\begin{align*}
\mathcal{L}\bar{V}\left( \mathbf{z} \right) & \leq 2\left( {p_{3}z_{3} + p_{2}z_{1}} \right)u_{1}^{\prime} + 2\left( {p_{3}z_{4} + p_{2}z_{2}} \right)u_{2}^{\prime} + \left( {2p_{1} - p_{2}\sigma_{3}^{2}} \right)\left( {z_{1}z_{3} + z_{2}z_{4}} \right) + 2p_{2}\left( {z_{3}^{2} + z_{4}^{2}} \right) + p_{1}\left( {\sigma_{1}^{2} + \sigma_{2}^{2}} \right)\frac{z_{1}^{2} + z_{2}^{2}}{r_{g}^{2}} \\
& = \left( {2p_{1} - p_{2}\sigma_{3}^{2}} \right)\left( {z_{1}z_{3} + z_{2}z_{4}} \right) + 2p_{2}\left( {z_{3}^{2} + z_{4}^{2}} \right) + p_{1}\left( {\sigma_{1}^{2} + \sigma_{2}^{2}} \right)\frac{z_{1}^{2} + z_{2}^{2}}{r_{g}^{2}} \\ & = \left\lbrack {- \frac{p_{2}}{2p_{3}}\left( {2p_{1} - p_{2}\sigma_{3}^{2}} \right) + \frac{p_{1}\left( {\sigma_{1}^{2} + \sigma_{2}^{2}} \right)}{r_{g}^{2}}} \right\rbrack\left( {z_{1}^{2} + z_{2}^{2}} \right) + \left( {- \frac{p_{1}p_{3}}{p_{2}} + 2p_{2} + \frac{1}{2}p_{3}\sigma_{3}^{2}} \right)\left( {z_{3}^{2} + z_{4}^{2}} \right)
\end{align*}
Therefore, if the following inequalities hold, $\mathcal{L}\bar{V}\left( \mathbf{z} \right) < 0$ when $L_{B}\bar{V}\left( \mathbf{z} \right) = 0$, and then $\bar{V}\left( \mathbf{z} \right)$ becomes an unconstrained control Lyapunov function for the linearized system.
\begin{eqnarray*}
p_{1},p_{2},p_{3} > 0,p_{1}p_{3} - p_{2}^{2} > 0
,
\frac{p_{2}}{p_{3}} - \frac{\sigma_{1}^{2} + \sigma_{2}^{2}}{r_{g}^{2}} > 0
,
p_{1} > {\max\left( {\frac{p_{2}^{2}r_{g}^{2}\sigma_{3}^{2}}{2p_{2}r_{g}^{2} - 2p_{3}\left( \sigma_{1}^{2} + \sigma_{2}^{2} \right)},\frac{2p_{2}^{2}}{p_{3}} + \frac{1}{2}p_{2}\sigma_{3}^{2}} \right)}
\end{eqnarray*}
Moreover, $L_{B}\bar{V}\left( \mathbf{z} \right) = 0$ implies $
p_{2}\left( {x\cos\theta + y\sin\theta} \right) + p_{3}v = 0$ and $y\cos\theta - x\sin\theta = 0$, and hence
\begin{align*}
\mathcal{L}\bar{V}\left( \bar{\mathbf{x}} \right) & = 2p_{2}\left( {y\cos\theta - x\sin\theta} \right)v\omega + 2\left\lbrack {p_{2}\left( {x\cos\theta + y\sin\theta} \right)} \right\rbrack\dot{v} + 2\left\lbrack {p_{1}\left( {x\cos\theta + y\sin\theta} \right) + p_{2}v} \right\rbrack v \\
& + p_{1}\left( {\sigma_{1}^{2} + \sigma_{2}^{2}} \right)\left( {1 - d_{Xg} + d_{Xg}\frac{x^{2} + y^{2}}{r_{g}^{2}}} \right) - p_{2}\left( {x\cos\theta + y\sin\theta} \right)v\sigma_{3}^{2} \\
& = \left( {2p_{1} - p_{2}\sigma_{3}^{2}} \right)\left( {x\cos\theta + y\sin\theta} \right)v + 2p_{2}v^{2} + p_{1}\left( {\sigma_{1}^{2} + \sigma_{2}^{2}} \right)\left( {1 - d_{Xg} + d_{Xg}\frac{x^{2} + y^{2}}{r_{g}^{2}}} \right) \\ 
& \leq \left( {2p_{1} - p_{2}\sigma_{3}^{2}} \right)\left( {x\cos\theta + y\sin\theta} \right)v + 2p_{2}v^{2} + p_{1}\left( {\sigma_{1}^{2} + \sigma_{2}^{2}} \right)\frac{x^{2} + y^{2}}{r_{g}^{2}} \\
& = \left\lbrack {- \frac{p_{2}}{2p_{3}}\left( {2p_{1} - p_{2}\sigma_{3}^{2}} \right) + \frac{p_{1}\left( {\sigma_{1}^{2} + \sigma_{2}^{2}} \right)}{r_{g}^{2}}} \right\rbrack\left( {x^{2} + y^{2}} \right) + \left( {- \frac{p_{1}p_{3}}{p_{2}} + 2p_{2} + \frac{1}{2}p_{3}\sigma_{3}^{2}} \right)v^{2} < 0
\end{align*}
which means $\bar{V}\left( \bar{\mathbf{x}} \right)$ is an unconstrained control Lyapunov function for the augmented system. We have
\begin{eqnarray*}
\mathcal{L}\bar{V}\left( \bar{\mathbf{x}} \right) = p_{1}\left( {\sigma_{1}^{2} + \sigma_{2}^{2}} \right)\left( {1 - d_{Xg} + d_{Xg}\frac{x^{2} + y^{2}}{r_{g}^{2}}} \right) - p_{2}\left( {x\cos\theta + y\sin\theta} \right)v\sigma_{3}^{2}
\end{eqnarray*}
when $L_{\bar{g}}\bar{V}\left( \bar{\mathbf{x}} \right) = 0$. Finally, if $L_{\mathbf{g}}V\left( \mathbf{x} \right) = 0$,
\begin{align*}
\mathcal{L}V\left( \mathbf{x} \right) & = \left\lbrack {\left( {p_{1} - \frac{p_{2}^{2}}{p_{3}}{\cos}^{2}\theta} \right)\sigma_{1}^{2} + \left( {p_{1} - \frac{p_{2}^{2}}{p_{3}}{\sin}^{2}\theta} \right)\sigma_{2}^{2}} \right\rbrack\left( {1 - d_{Xg} + d_{Xg}\frac{x^{2} + y^{2}}{r_{g}^{2}}} \right) \\
& + \frac{p_{2}^{2}}{p_{3}}\left\lbrack {\left( {x\cos\theta + y\sin\theta} \right)^{2} - \left( {y\cos\theta - x\sin\theta} \right)^{2}} \right\rbrack\sigma_{3}^{2} \\
& < p_{1}\left( {\sigma_{1}^{2} + \sigma_{2}^{2}} \right)\left( {1 - d_{Xg} + d_{Xg}\frac{x^{2} + y^{2}}{r_{g}^{2}}} \right) + \frac{p_{2}^{2}}{p_{3}}\left( {x\cos\theta + y\sin\theta} \right)^{2} = \mathcal{L}\bar{V}\left( \bar{\mathbf{x}} \right)\left. \hspace{0pt} \right|_{p = p_{0},L_{\bar{g}}\bar{V}{(\bar{\mathbf{x}})} = 0} < 0
\end{align*}
and consequently $
V\left( \mathbf{x} \right)$ is an unconstrained control Lyapunov function for the original system. We can obtain $
V\left( \mathbf{x} \right)$ according to (\ref{eq19}), namely,
\begin{eqnarray*}
V\left( \mathbf{x} \right) = p_{1}\left( x^{2} + y^{2} \right) - \frac{p_{2}^{2}}{p_{3}}\left( {x\cos\theta + y\sin\theta} \right)^{2}
\end{eqnarray*}
The proposed stochastic CLBF is
\begin{eqnarray*}
W_{c}\left( \mathbf{x} \right) = V\left( \mathbf{x} \right) + {\sum\limits_{i = 1}^{n_{B}}{\lambda_{i}B_{i}\left( \mathbf{x} \right)}} + \kappa = p_{1}\left( x^{2} + y^{2} \right) - \frac{p_{2}^{2}}{p_{3}}\left( {x\cos\theta + y\sin\theta} \right)^{2} + {\sum\limits_{i = 1}^{n_{B}}{\lambda_{i}\left\lbrack {d_{Xi}\frac{B_{imax} - B_{imin}}{1 + e^{- \frac{k_{Bi}(l_{Di} - F_{Bi})}{F_{Bi}{({l_{Xi} - F_{Bi}})}}}~} + B_{imin}} \right\rbrack}} + \kappa
\end{eqnarray*}
We then compute the feasible regions $X_\phi$ and $X_L$. Quadruples $\left( {x_{i\_ obs},y_{i\_ obs},l_{Di},l_{Xi}} \right)$ of obstacles are $\left(30,25,25,36\right)$, $\left(50,50,25,36\right)$, $\left(68,30,56.25,90.25\right)$, $\left(80,60,56.25,90.25\right)$, respectively. Other parameters in the simulations are displayed in Table 1.

\begin{center}
\begin{table}[ht]
\centering
\caption{Parameters in the simulations\label{tab1}}%
\begin{tabular*}{363pt}{c|cccc}
\hline   \multirow{2}*{\textbf{Simulation conditions}} & $\sigma_{1} = 0.3$ & $\sigma_{2} = 0.3$ & $\sigma_{3} = 0.6$ & $r_{g} = 5$  \\  
 & \multicolumn{4}{c}{$\mathbf{x}_{0} = \left( {100,80, - \frac{\pi}{2}} \right)$} \\  
\hline   \textbf{CLF} & $p_1=10$ & $p_2=1$ & $p_3=1$  & \\ 
\hline   \textbf{CBF} & \multicolumn{2}{c}{$B_{min} = \lbrack - 10, - 10, - 10, - 10\rbrack$} & \multicolumn{2}{c}{$B_{max} = \lbrack 15,15,15,15\rbrack$}  \\  
\hline   \multirow{6}*{\textbf{CLBF}} & \multicolumn{2}{c}{$c_{1} = p_{1} - \frac{3p_{2}^{2}}{2p_{3}} = \frac{17}{2}$} & \multicolumn{2}{c}{$\lambda_{i} = \frac{c_{2}c_{3i} - c_{1}c_{4i}}{\eta_{i}} + K_{\lambda i}$}   \\  
 & \multicolumn{2}{c}{$c_{2} = p_{1} = 3$} & \multicolumn{2}{c}{$\eta_{i} = - B_{min\_ i}$}  \\
 & \multicolumn{4}{c}{$
c_{3i} = \begin{matrix}
{\max} \\
{x \in \partial X_{i}} \\
\end{matrix}\left( {x^{2} + y^{2}} \right) = \left( {\sqrt{x_{i\_ obs}^{2} + y_{i\_ obs}^{2}} + \sqrt{l_{Xi}}} \right)^{2}$}  \\
 & \multicolumn{4}{c}{$c_{4i} = \begin{matrix}
{\min} \\
{x \in D_{i}} \\
\end{matrix}\left( {x^{2} + y^{2}} \right) = \left( {\sqrt{x_{i\_ obs}^{2} + y_{i\_ obs}^{2}} - \sqrt{l_{Di}}} \right)^{2}$}  \\
 & \multicolumn{4}{c}{$K_{\lambda} = \lbrack 100000,80000,100000,80000\rbrack$} \\
 & \multicolumn{4}{c}{$
\kappa = \frac{1}{2}\left\lbrack {\underset{i}{\max}\left( {{\sum\limits_{j = 1,j \neq i}^{n_{B}}{\lambda_{j}\eta_{j}}} - c_{1}c_{4i}} \right) + {\sum\limits_{i = 1}^{n_{B}}{\lambda_{i}\eta_{i}}} - c_{2}\underset{i}{\max}\left( c_{3i} \right)} \right\rbrack$}   \\
\hline   \multirow{2}*{\textbf{MPC}} & \multicolumn{2}{c}{$\mathbf{Q} = diag(10,10,0)$} & \multicolumn{2}{c}{$\mathbf{R} = diag(0,0,0)$}   \\
 & \multicolumn{2}{c}{$T = 0.1$} & \multicolumn{2}{c}{$N=20$} \\
\hline
\end{tabular*}
\end{table}
\end{center}

The decay rate $k_{Bi}$ of CBF largely influences the feasible regions. Near the boundary of $X_i$, the term $k_{Bi}^2$ in $\frac{\partial^{2}B_{i}}{\partial\mathbf{x}^{2}}$ mainly dominates the magnitudes of $\frac{1}{2}\mathrm{tr}\left( {\mathbf{\sigma}^{T}\left( \mathbf{x} \right)\frac{\partial^{2}W_{c}}{\partial\mathbf{x}^{2}}\mathbf{\sigma}\left( \mathbf{x} \right)} \right)$ in $\mathcal{L}W_{c}\left( \mathbf{x} \right)$, implying that a small $k_{Bi}$ is preferred. However, on the other hand, constraints (\ref{eq29}) are not continuously implemented in the optimization, therefore gradients of $B_i$ should be uniform enough to counteract the effects of $\mathcal{L}V\left( \mathbf{x} \right)$. We show $k_{B4}$’s influence on $B_4$ in Figure \ref{fig1}. 

\begin{figure}[htbp]
\centerline{\includegraphics[width=300pt,height=16pc]{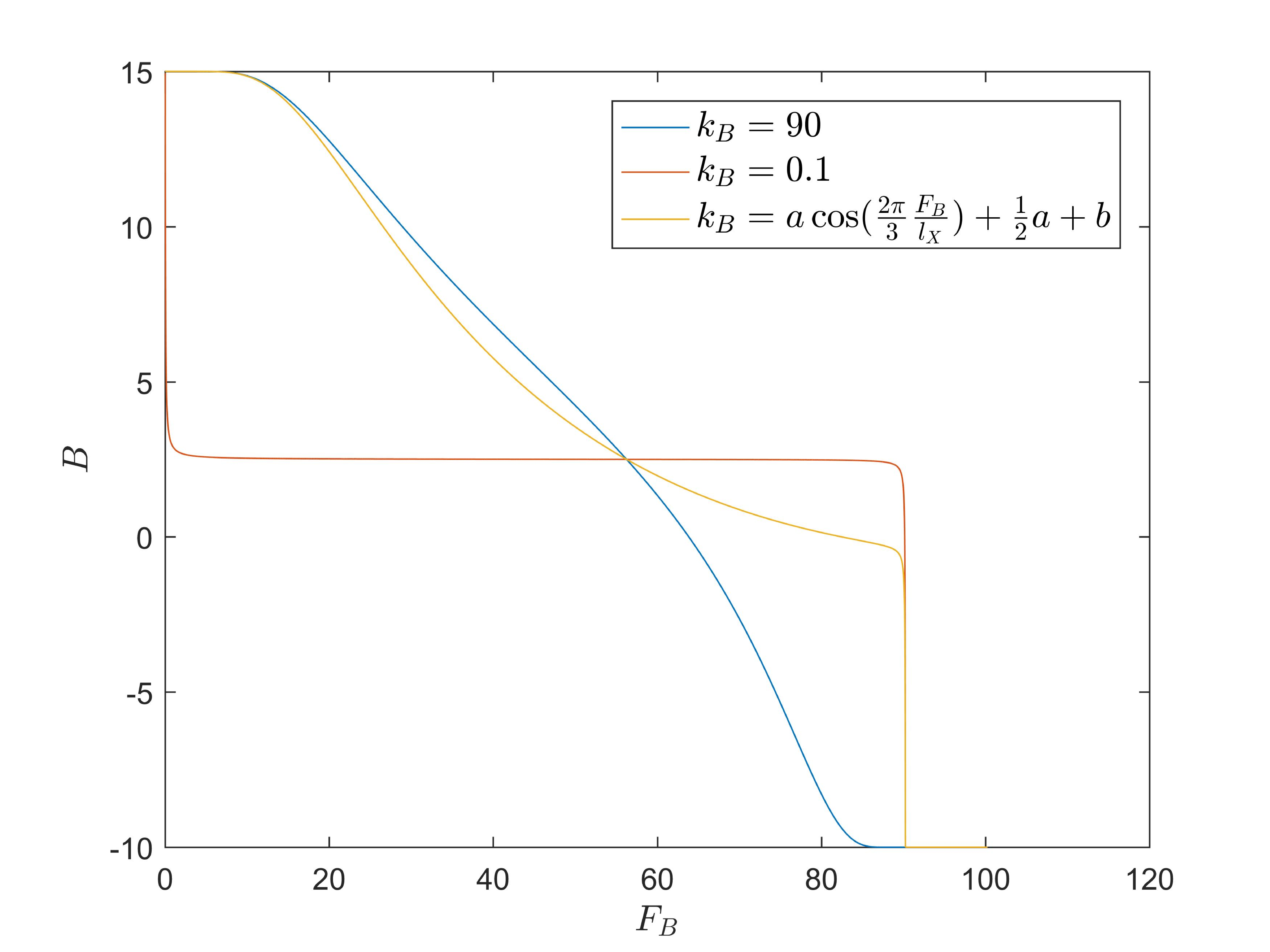}}
\caption{$k_{Bi}$’s influence on $B_i$\label{fig1}}
\end{figure}

\begin{figure}[htbp]
\centering
{\includegraphics[width=300pt,height=16pc]{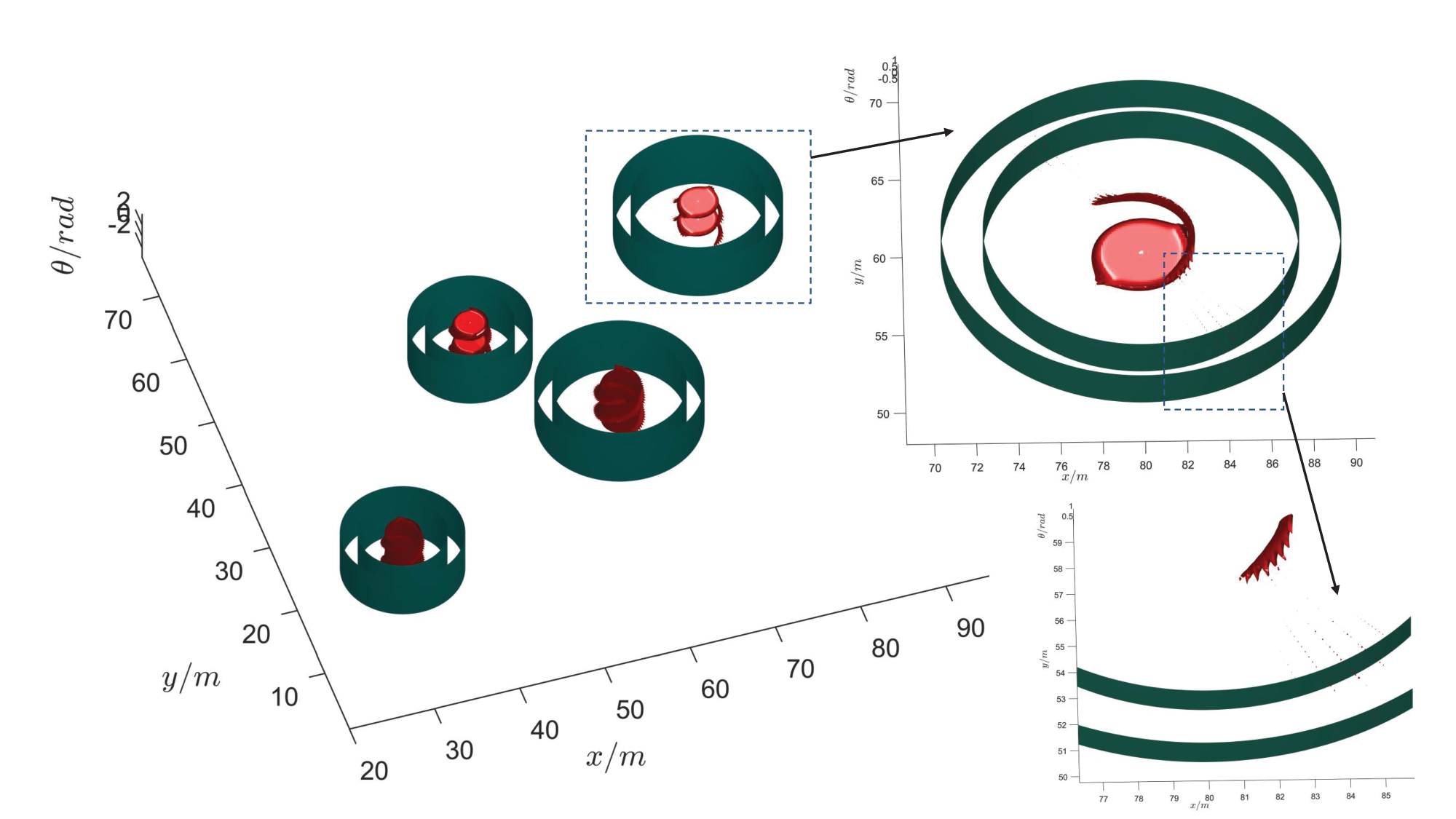}}

{\includegraphics[width=300pt,height=16pc]{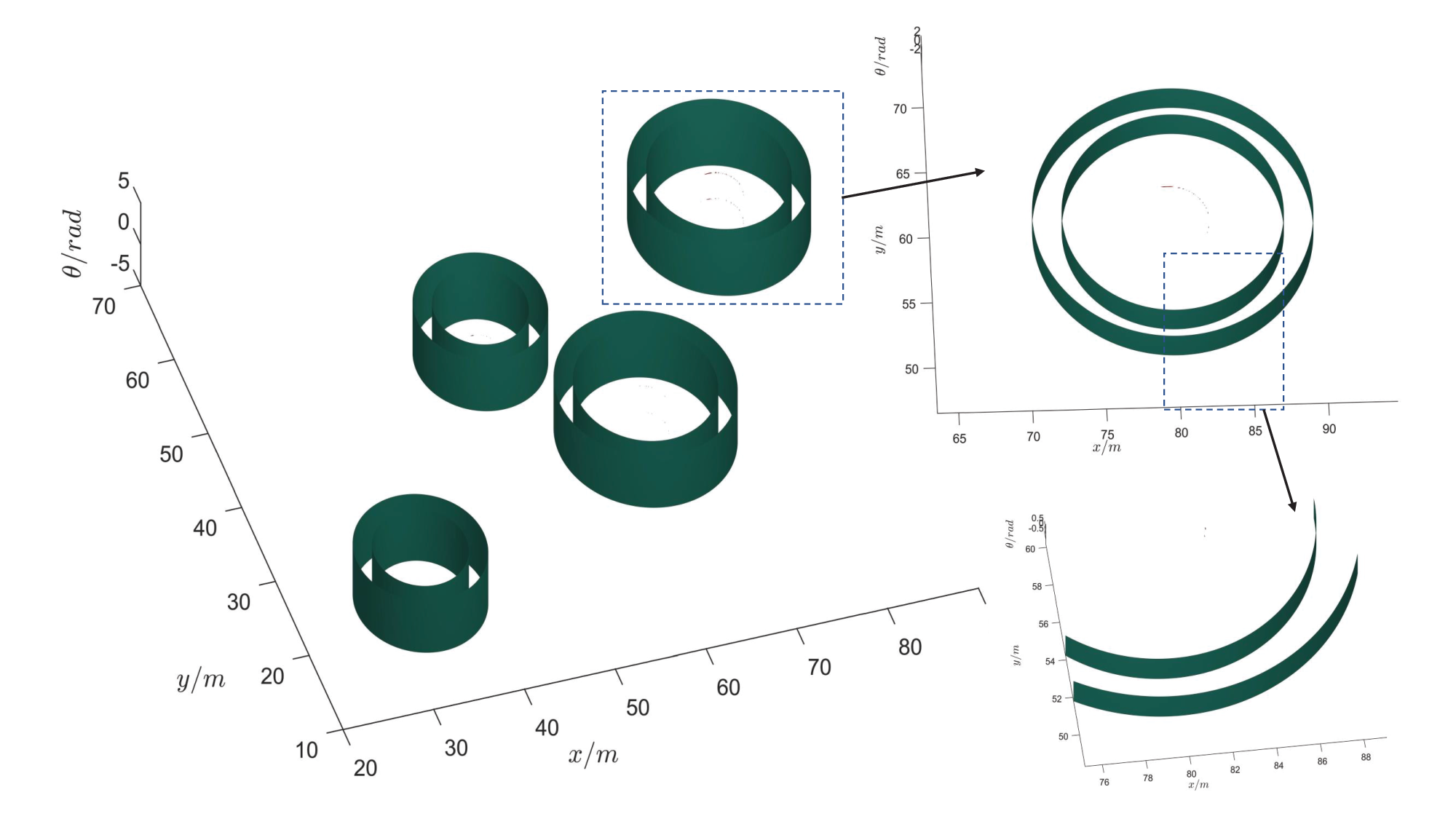} }
\caption{The feasible regions $X_\phi$ with $\rho=0.005$ (the former) and $X_L$ (the later). The green surfaces are boundaries of $X_i$ and $D_i$ while the red region is outside $X_\phi$ or $X_L$.\label{fig2}}
\end{figure}

Figure \ref{fig1} indicates small $k_{Bi}$ will cause a sharp rise after crossing the boundary of $X_i$ and remain nearly constant as $F_{Bi}\left( \mathbf{x} \right)$ approaches to zero, therefore $k_{Bi}\left( F_{Bi} \right)$ should be a monotonically decreasing function. We construct $k_{Bi} = {a_{i}{\cos\left( {\frac{2\pi}{3}\frac{F_{Bi}}{l_{Xi}}} \right)}} + \frac{1}{2}a_{i} + b_{i} \in \left\lbrack {b_{i},\frac{3}{2}a_{i} + b_{i}} \right\rbrack$ with $a_i=60,b_i=0.1$ and the feasible regions are given by Figure \ref{fig2}. It can be seen that, near unsafe regions $D_i$, there are some small regions that the CLBF-based candidate controller is not a feasible solution while MPC can still provide admissible control inputs. 

Our simulations are based on CasADi\cite{Andersson2019}. The wheeled mobile robot starts from $\left(x,y\right)=\left(100,80\right)$ and runs towards the origin $\left(x,y\right)=\left(0,0\right)$. 
We run the simulation 20 times under same settings and the results are displayed in Figures \ref{fig3}.
To show benefits of proposed methods, in the left graph, we ignore the influence of auxiliary Lyapunov controller 
$ \mathcal{L}W_{c}\left( {\mathbf{x}\left( i \middle| k \right),\mathbf{\phi}\left( \mathbf{x}\left( i \middle| k \right) \right)} \right)$ in (\ref{eq29}) and 
 event-triggering mechanisms in (\ref{eq38}). When the wheeled mobile robot enters $X_i$ (blue circles), CBFs work and impose constraints on MPC to steer it away from unsafe regions $D_i$ (red circles). States and inputs are displayed in Figures \ref{fig4}. As sample interval $T$ approaches zero, the states $x$ and $y$ are stochastically asymptotically stabilizable and converge to zero with constrained inputs. The risk of entering $D_{relaxed}$ will finally tend to zero, and so is that of unsafe regions $D_i$. 

\begin{figure}[htbp]
		\begin{minipage}{0.44\textwidth}
			\centering
			\includegraphics[width=250pt,height=20pc]{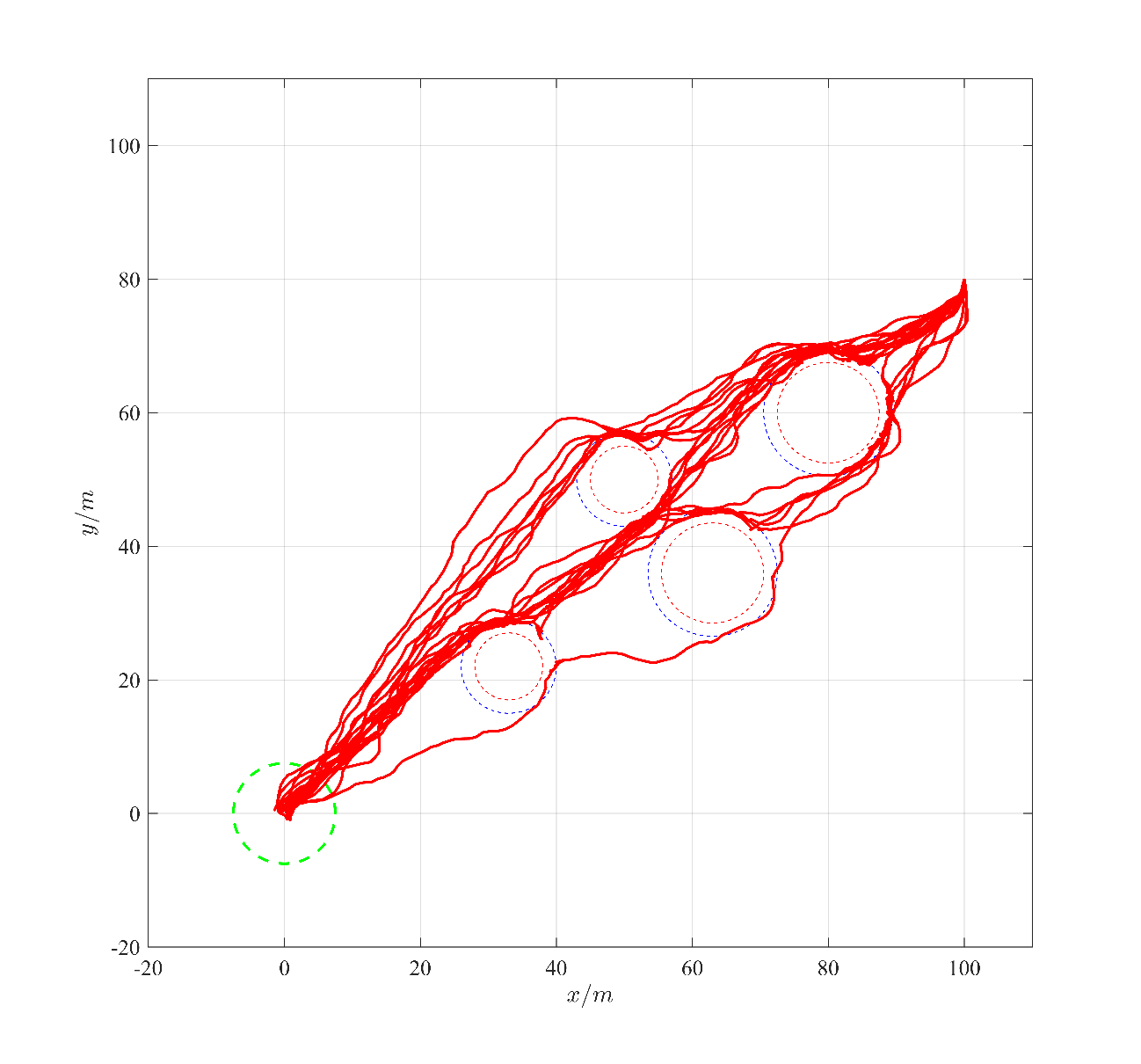}
		\end{minipage}
		\qquad
		\begin{minipage}{0.5\textwidth}
			\centering
			\includegraphics[width=250pt,height=20pc]{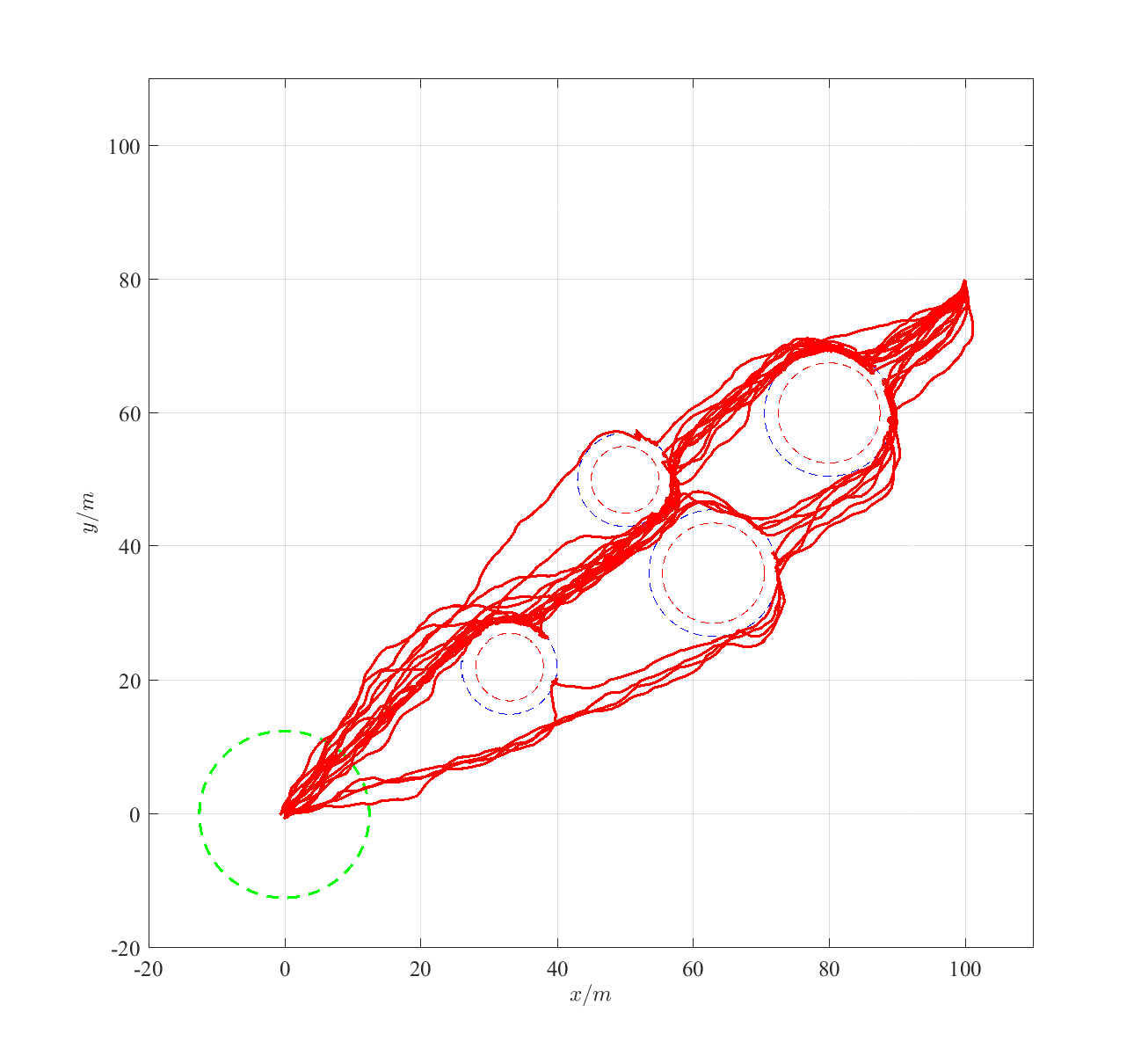}
		\end{minipage}
  \caption{20 simulations with same conditions. Red, blue and green circles denote $X_i$, $D_i$ and $X_g$ respectively. The right one considers auxiliary Lyapunov controller and event-triggering mechanisms while the left does not.}\label{fig3}
\end{figure}

\begin{figure}[htbp]
		\begin{minipage}{0.44\textwidth}
			\centering
			\includegraphics[width=230pt,height=17pc]{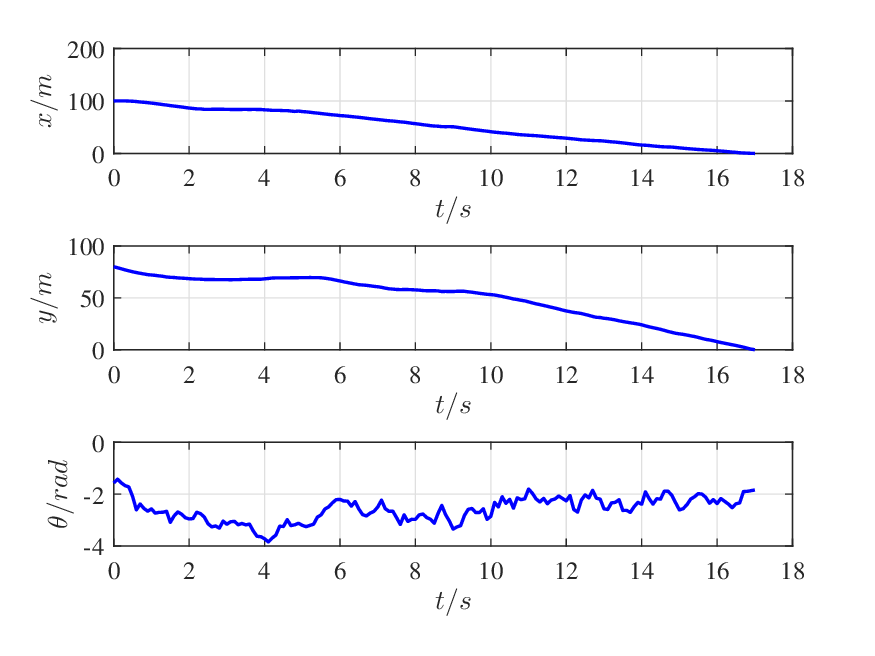}
		\end{minipage}
		\qquad
		\begin{minipage}{0.5\textwidth}
			\centering
			\includegraphics[width=230pt,height=17pc]{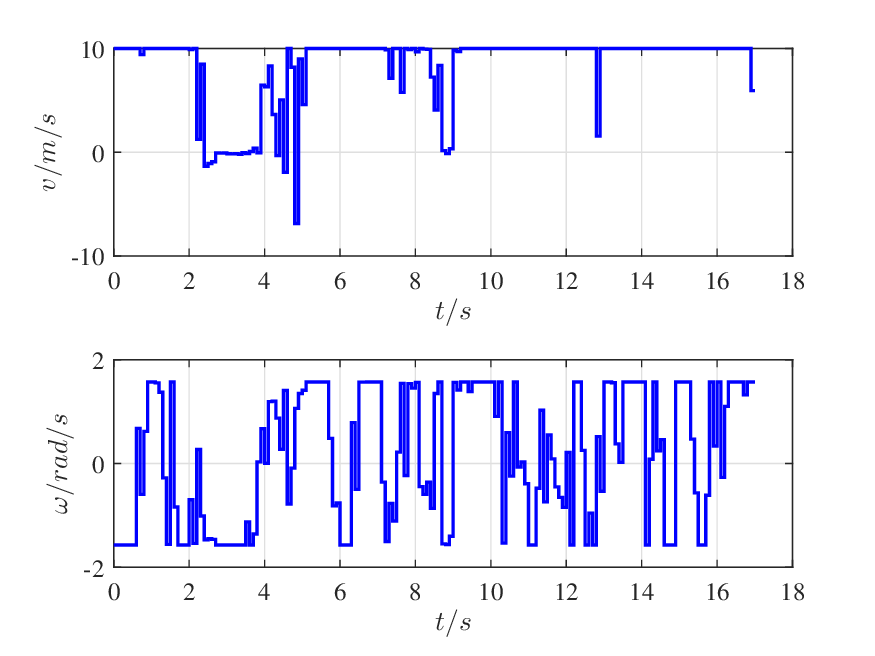}
		\end{minipage}
\caption{States and control inputs of the wheeled mobile robot}\label{fig4}
\end{figure}

\begin{figure}[htbp]
		\begin{minipage}{0.44\textwidth}
			\centering
			\includegraphics[width=230pt,height=14pc]{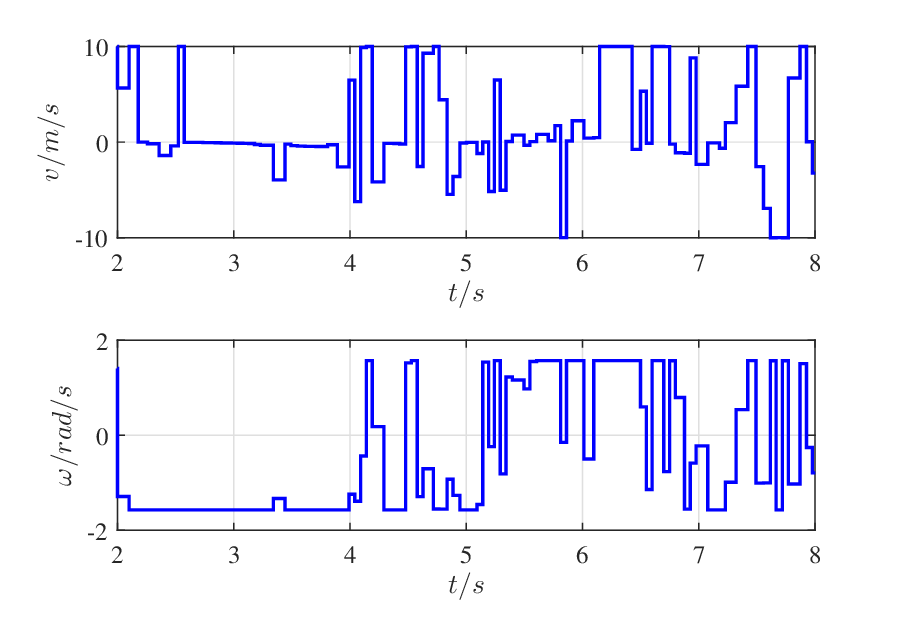}
		\end{minipage}
		\qquad
		\begin{minipage}{0.5\textwidth}
			\centering
			\includegraphics[width=230pt,height=14pc]{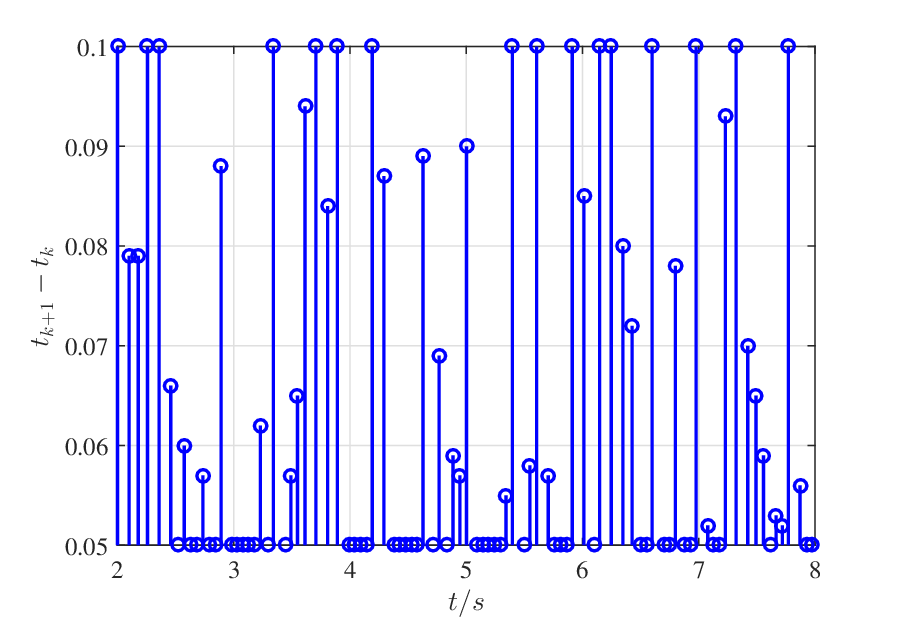}
		\end{minipage}
\caption{Control inputs and execution times under event-triggering mechanisms}\label{fig5}
\end{figure}

Auxiliary Lyapunov controller mainly influences control performance during sampling intervals. On one hand, $ \mathcal{L}W_{c}\left( {\mathbf{x}\left( i \middle| k \right),\mathbf{\phi}\left( \mathbf{x}\left( i \middle| k \right) \right)} \right)$ in (\ref{eq29}) restricts convergence rate of states, on the another hand, it provides some robustness against sampling errors $\mathbf{e}_{t}$, which may invalidate CLBF constraints, between the last sample-hold states $\mathbf{x}_{t_{k}}$ and current states $\mathbf{x}_{t}$. Similarly, the triggering condition (\ref{eq38}) recovers CLBF constraints after $t_{k} + \tau$, therefore enlarges maximum allowable sampling period and improves effects of obstacle avoidance, as what can be seen in the right graph.
To balance computation expense and theoretical stability guarantee (although the estimation in Lemma \ref{lem1} may be a little large), a proper minimum interexecution time $\tau$ is essential. 
We modify some parameters to $p_{1}=0.1,p_{2}=p_{3}=0.001,r_g=10$ and obtain $L_1=0,L_2=0.0424,L_3=1.02,L_4=1,\mu_{1} = \rho c_{1} - 4\left( {p_{1} - \frac{p_{2}^2}{p_{3}}} \right)^{2} = 0.3076,\mu_{2} = 4\left( {p_{1} - \frac{p_{2}^2}{p_{3}}} \right)^{2} = 0.0324,\mu_{3} = 0,
\mu_{4} = 2\left( {p_{1} - \frac{p_{2}^2}{p_{3}}} \right) = 0.18, T^{*}=0.0525$ according to Lemma \ref{lem1}. 
Let $\tau=0.05,\tau_{max}=0.1,T=0.05$ and simulations results are displayed in the right graph. Note that the proposed MPC is always feasible $\forall\mathbf{x} \in X_{g}$ if auxiliary Lyapunov controller is taken as a candidate solution and hence closed-loop systems is ultimately bounded in the mean square.
Control inputs and execution times under event-triggering mechanisms are displayed in Figure \ref{fig5}. 
It can be seen that, with the proposed CLBF-based stochastic MPC, the wheeled mobile robot is capable to reach the goal set while avoiding unsafe sets.

\section{Conclusion}
This paper presents a continuous-time stochastic MPC approach for nonlinear systems. We give the notion of stochastic CLBF and its construction. Moreover, a stochastic CLF design method for differentially flat systems via dynamic feedback linearization is presented. The proposed stochastic MPC can ensure stability, safety and feasibility and is validated through a collision avoidance simulation of wheeled mobile robots. Future work focuses on theoretical analysis of safety in sampling intervals.

\nocite{*}
\bibliography{CLBF-SMPC.bib}

\end{document}